\documentclass{article} 

\usepackage[final]{nips_2017}

\usepackage{multicol}
\usepackage{multirow}
\usepackage{algorithm}
\usepackage{algpseudocode} 
\usepackage{amsmath}
\usepackage{amsthm}
\usepackage{amssymb}
\usepackage{natbib}
\usepackage{rotating}
\usepackage{thmtools, thm-restate}
\usepackage{tikz}
\usepackage{url}

\usetikzlibrary{shapes,snakes}

\newcommand\Acomm[1]{\State {\color{antiquebrass}{// {#1}}}}

\newtheorem{theorem}{Theorem}[section]

\algdef{SE}[DOWHILE]{Do}{doWhile}{\algorithmicdo}[1]{\algorithmicwhile\ #1}%

\newcommand{\CLARANS}{\texttt{clarans}}
\newcommand{\clarans}{\texttt{clarans}}

\newcommand{\clara}{\texttt{clara}}
\newcommand{\lloyd}{\texttt{lloyd}}
\newcommand{\kmeanspp}{\texttt{km++}}

\newcommand{\pam}{\texttt{pam}}
\newcommand{\vik}{\texttt{vik}}
\newcommand{\BF}{\texttt{bf}}
\newcommand{\uni}{\texttt{uni}}
\newcommand{\medlloyd}{\texttt{medlloyd}}
\newcommand{\mcmc}{\texttt{afk-mc$^{2}$}}

\DeclareMathOperator*{\argmin}{arg\,min}

\DeclareMathOperator*{\dist}{dist}

\DeclareMathOperator*{\energy}{\psi}
\DeclareMathOperator*{\Energy}{\psi}

\definecolor{antiquebrass}{rgb}{0.44, 0.47, 0.77}

\title{K-Medoids for K-Means Seeding}

\author{
James Newling\\
Idiap Research Institue and \\
\'{E}cole polytechnique f\'{e}d\'{e}rale de Lausanne \\
\texttt{james.newling@idiap.ch} \\
\And
Fran\c{c}ois Fleuret \\
Idiap Research Institue and \\
\'{E}cole polytechnique f\'{e}d\'{e}rale de Lausanne \\
\texttt{francois.fleuret@idiap.ch} \\
}

\begin{document} 

\maketitle

\begin{abstract}

We show experimentally that the algorithm \clarans{} of \citet{clarans} finds better $K$-medoids solutions than the Voronoi iteration algorithm of \citet{the_elements}. This finding, along with the similarity between the Voronoi iteration algorithm and Lloyd's $K$-means algorithm, motivates us to use \CLARANS{} as a $K$-means initializer. We show that \CLARANS{} outperforms other algorithms on 23/23 datasets with a mean decrease over \texttt{k-means-++}~\citep{arthur_2007_kmeanspp} of $30 \%$ for initialization mean squared error (MSE) and $3 \%$ for final MSE. We introduce algorithmic improvements to \clarans{} which improve its complexity and runtime, making it a viable initialization scheme for large datasets. 

\end{abstract}

\section{Introduction}

\subsection{$K$-means and $K$-medoids}
The $K$-means problem is to find a partitioning of points, so as to minimize the sum of the squares of the distances from points to their assigned partition's mean. In general this problem is NP-hard, and in practice approximation algorithms are used. The most popular of these is Lloyd's algorithm, henceforth \lloyd{}, which alternates between freezing centers and assignments, while updating the other.  Specifically, in the \textit{assignment} step, for each point the nearest (frozen) center is determined. Then during the \textit{update} step, each center is set to the mean of points assigned to it.  \texttt{lloyd} has applications in data compression, data classification, density estimation and many other areas, and was recognised in~\citet{wu2008wu} as one of the top-10 algorithms in data mining.

The closely related $K$-medoids problem differs in that the center of a cluster is its medoid, not its mean, where the medoid is the cluster member which minimizes the sum of \textit{dissimilarities} between itself and other cluster members. In this paper, as our application is $K$-means initialization, we focus on the case where dissimilarity is squared distance, although $K$-medoids generalizes to non-metric spaces and arbitrary dissimilarity measures, as discussed in \S\ref{app::generalized}.

By modifying the update step in \lloyd{} to compute medoids instead of means, a viable $K$-medoids algorithm is obtained. This algorithm has been proposed at least twice~\citep{the_elements, park_2009_kmedoids} and is often referred to as the Voronoi iteration algorithm. We refer to it as \medlloyd{}. 

Another $K$-medoids algorithm is \clarans{} of \citet{clarans, clarans_other}, for which there is no direct $K$-means equivalent. It works by randomly proposing swaps between medoids and non-medoids, accepting only those which decrease MSE. We will discuss how \clarans{} works, what advantages it has over \medlloyd{}, and our motivation for using it for $K$-means initialization in \S\ref{sec::kmedoids} and \S\ref{app::generalized}. 

\subsection{$K$-means initialization}

\lloyd{} is a \textit{local} algorithm, in that far removed centers and points do not directly influence each other. This property contributes to \lloyd{}'s tendency to terminate in poor minima if not well initialized. Good initialization is key to guaranteeing that the refinement performed by \lloyd{} is done in the vicinity of a good solution, an example showing this is given in Figure~\ref{initex}.

\begin{figure}
\begin{center}
\tikzstyle{position}=[]
\begin{tikzpicture}[scale=.69]
\path node (x1) at (0,0.001) [position] {${\color{orange}{\star}}$}
      node (x2) at (2,0.001) [position] {${\color{orange}{\bigcirc}}$}
      node (x2) at (2,0.001) [position] {${\color{orange}{\star}}$}
      node (x3) at (3.3,0.001) [position] {${\color{blue}{\bigcirc}}$}
      node (x3) at (3.3,0.001) [position] {${\color{blue}{\star}}$};
\node [anchor=west, yshift = 0mm, xshift = -18mm] at (x1) {initial:};

\path node (x1) at (0,-0.523) [position] {${\color{orange}{\star}}$}
      node (x2) at (1.0,-0.523) [position] {${\color{orange}{\bigcirc}}$}
      node (x2) at (2,-0.523) [position] {${\color{orange}{\star}}$}
      node (x3) at (3.3,-0.523) [position] {${\color{blue}{\bigcirc}}$}
      node (x3) at (3.3,-0.523) [position] {${\color{blue}{\star}}$};
\node [anchor=west, yshift = 0mm, xshift = -18mm] at (x1) {final:};

\path node (x1) at (0 + 11.1112, 0.001) [position] {${\color{orange}{\star}}$}
      node (x2) at (0 + 11.1112, 0.001) [position] {${\color{orange}{\bigcirc}}$}
      node (x2) at (2 + 11.1112, 0.001) [position] {${\color{blue}{\star}}$}
      node (x3) at (3.3 + 11.1112, 0.001) [position] {${\color{blue}{\bigcirc}}$}
      node (x3) at (3.3 + 11.1112, 0.001) [position] {${\color{blue}{\star}}$};
\node [anchor=west, yshift = 0mm, xshift = -18mm] at (x1) {initial:};

\path node (x1) at (0 + 11.1112, -0.523) [position] {${\color{orange}{\star}}$}
      node (x2) at (0.0 + 11.1112, -0.523) [position] {${\color{orange}{\bigcirc}}$}
      node (x2) at (2 + 11.1112, -0.523) [position] {${\color{blue}{\star}}$}
      node (x3) at (2.65 + 11.1112, -0.523) [position] {${\color{blue}{\bigcirc}}$}
      node (x3) at (3.3 + 11.1112, -0.523) [position] {${\color{blue}{\star}}$};
\node [anchor=west, yshift = 0mm, xshift = -18mm] at (x1) {final:};
\end{tikzpicture}
\caption{$N=3$ points, to be partitioned into $K=2$ clusters with \lloyd{}, with two possible initializations (top) and their solutions (bottom). Colors denote clusters, stars denote samples, rings denote means. Initialization with \clarans{} enables jumping between the initializations on the left and right, ensuring that when \lloyd{} eventually runs it avoids the local minimum on the left.} 
\label{initex}
\end{center}
\vspace{-1em}

\end{figure}
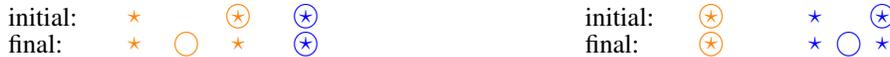

In the comparative study of $K$-means initialization methods of~\citet{celebi_2013_initialisation}, 8 schemes are tested across a wide range of datasets. Comparison is done in terms of speed (time to run initialization+\lloyd{}) and energy (final MSE). They find that 3/8 schemes should be avoided, due to poor performance. One of these schemes is uniform initialization, henceforth \uni{}, where $K$ samples are randomly selected to initialize centers. Of the remaining 5/8 schemes, there is no clear best, with results varying across datasets, but the authors suggest that the algorithm of \citet{bradley_fayyad}, henceforth \BF{}, is a good choice. 

The \BF{} scheme of~\citet{bradley_fayyad} works as follows. Samples are separated into $J$ $( = 10 )$ partitions. \lloyd{} with \uni{} initialization is performed on each of the partitions, providing $J$ centroid sets of size $K$. A superset of $JK$ elements is created by concatenating the $J$ center sets. \lloyd{} is then run $J$ times on the superset, initialized at each run with a distinct center set. The center set which obtains the lowest MSE on the superset is taken as the final initializer for the final run of \lloyd{} on all $N$ samples. 

Probably the most widely implemented initialization scheme other than \uni{} is \texttt{k-means++}~\citep{arthur_2007_kmeanspp}, henceforth \kmeanspp{}. Its popularity stems from its simplicity, low computational complexity, theoretical guarantees, and strong experimental support. The algorithm works by sequentially selecting $K$ seeding samples. At each iteration, a sample is selected with probability proportional to the square of its distance to the nearest previously selected sample. 

The work of~\citet{bachem16tortoise} focused on developing sampling schemes to accelerate \kmeanspp{}, while maintaining its theoretical guarantees. Their algorithm \mcmc{} results in as good initializations as \kmeanspp{}, while using only a small fraction of the $KN$ distance calculations required by \kmeanspp{}. This reduction is important for massive datasets. 

In none of the 4 schemes discussed is a center ever replaced once selected. Such refinement is only performed during the running of \lloyd{}. In this paper we show that performing refinement during initialization with \clarans{}, before the final \lloyd{} refinement, significantly lowers $K$-means MSEs.

\subsection{Our contribution and paper summary}
We compare the $K$-medoids algorithms \clarans{} and \medlloyd{}, finding that \clarans{} finds better local minima, in \S\ref{sec::simstudy} and \S\ref{app::generalized}. We offer an explanation for this, which motivates the use of \clarans{} for initializing \lloyd{} (Figure~\ref{swappower}). We discuss the complexity of \clarans{}, and briefly show how it can be optimised in \S\ref{sec::complexity}, with a full presentation of acceleration techniques in \S\ref{app::clarans}.

Most significantly, we compare \clarans{} with methods \uni{}, \BF{}, \kmeanspp{} and \mcmc{} for $K$-means initialization, and show that it provides significant reductions in initialization and final MSEs in \S\ref{sec::mainresults}. We thus provide a conceptually simple initialization scheme which is demonstrably better than \kmeanspp{}, which has been the de facto initialization method for one decade now. 

Our source code at \url{https://github.com/idiap/zentas} is available under an open source license. It consists of a C++ library with Python interface, with several examples for diverse data types (sequence data, sparse and dense vectors), metrics (Levenshtein, $l_{1}$, etc.) and potentials (quadratic as in $K$-means, logarithmic, etc.).

\subsection{Other Related Works} 
Alternatives to \lloyd{} have been considered which resemble the swapping approach of \clarans{}. One is by~\cite{Hartigan_1975_book}, where points are randomly selected and reassigned. \cite{telgarsky_2010_hartigan} show how this heuristic can result in better clustering when there are few points per cluster.  

The work most similar to \clarans{} in the $K$-means setting is that of~\cite{Kanungo_2002_local}, where it is indirectly shown that \clarans{} finds a solution within a factor 25 of the optimal $K$-medoids clustering. The local search approximation algorithm they propose is a hybrid of \clarans{} and \lloyd{}, alternating between the two, with sampling from a kd-tree during the \clarans{}-like step. Their source code includes an implementation of an algorithm they call `Swap', which is exactly the \clarans{} algorithm of \cite{clarans}. 

\section{Two $K$-medoids algorithms}
\label{sec::kmedoids}
Like \kmeanspp{} and \mcmc{}, $K$-medoids generalizes beyond the standard $K$-means setting of Euclidean metric with quadratic potential, but we consider only the standard setting in the main body of this paper, referring the reader to~\ref{app::generalized} for a more general presentation. In Algorithm~\ref{alg::medlloyd}, \medlloyd{} is presented. It is essentially \lloyd{} with the update step modified for $K$-medoids.

\begin{multicols}{2}

\begin{algorithm}[H]
\begin{algorithmic}[1]
\input{medlloyd.alg}
\end{algorithmic}
\caption{two-step iterative \medlloyd{} algorithm (in vector space with quadratic potential).}
\label{alg::medlloyd}
\end{algorithm}

\columnbreak

\begin{algorithm}[H]
\begin{algorithmic}[1]
\input{clarans.alg}
\end{algorithmic}
\caption{swap-based \CLARANS{} algorithm (in a vector space and with quadratic potential).}
\label{alg::CLARANS101}
\end{algorithm}

\end{multicols}

In Algorithm~\ref{alg::CLARANS101}, \clarans{} is presented. Following a random initialization of the $K$ centers (line 2), it proceeds by repeatedly proposing a random swap (line 5) between a center $(i_-)$ and a non-center $(i_+)$. If a swap results in a reduction in energy (line 8), it is implemented (line 9). 
\clarans{}  terminates when $N_r$ consecutive proposals have been rejected. Alternative stopping criteria could be number of accepted swaps, rate of energy decrease or time.  We use $N_r = K^2$ throughout, as this makes proposals between all pairs of clusters probable, assuming balanced cluster sizes.

\clarans{} was not the first swap-based $K$-medoids algorithm, being preceded by \pam{} and \clara{} of \citet{generic_kmedoids}. It can however provide better complexity than other swap-based algorithms if certain optimisations are used, as discussed in \S\ref{sec::complexity}. 

When updating centers in \lloyd{} and \medlloyd{}, assignments are frozen. In contrast, with swap-based algorithms such as \clarans{}, assignments change along with the medoid index being changed ($i_-$ to $i_+$). As a consequence, swap-based algorithms look one step further ahead when computing MSEs, which helps them escape from the minima of \medlloyd{}. This is described in Figure~\ref{swappower}.

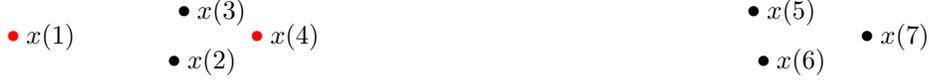
\begin{figure}
\center
\tikzstyle{position}=[]
\begin{tikzpicture}[scale=1.0]
\path node (c1) at (1.6221*0,0) [position] {$\color{red}{\bullet}$}
      node (c2) at (1.6221*2,0) [position] {$\color{red}{\bullet}$}
      node (x1) at (1.6221*1.4, .331211) [position] {$\bullet$}
      node (x2) at (1.6221*1.32, -.331211) [position] {$\bullet$}
      node (y1) at (1.6221*6.15, -.331211) [position] {$\bullet$}
      node (y2) at (1.6221*6.07, .331211) [position] {$\bullet$}
      node (y3) at (1.6221*7, 0) [position] {$\bullet$};

\node [xshift=5mm]              at (c1) {$x(1)$};
\node [xshift=5mm]              at (x2) {$x(2)$};
\node [xshift=5mm]              at (x1) {$x(3)$};
\node [xshift=5mm]              at (c2) {$x(4)$};
\node [xshift=5mm]              at (y2) {$x(5)$};
\node [xshift=5mm]              at (y1) {$x(6)$};
\node [xshift=5mm]              at (y3) {$x(7)$};
\end{tikzpicture}
\caption{Example with $N = 7$ samples, of which $K = 2$ are medoids. Current medoid indices are $1$ and $4$. Using \medlloyd{}, this is a local minimum, with final clusters $\{x(1)\}$, and the rest. \clarans{} may consider swap $(i_-,  i_+) = (4, 7)$ and so escape to a lower MSE. The key to swap-based algorithms is that cluster assignments are never frozen. Specifically, when considering the swap of $x(4)$ and $x(7)$, \clarans{} assigns $x(2), x(3)$ and $x(4)$ to the cluster of $x(1)$ \textit{before} computing the new MSE.}
\label{swappower}
\end{figure}

\section{A Simple Simulation Study for Illustration}
\label{sec::simstudy}

We generate simple 2-D data, and compare \medlloyd{}, \clarans{}, and baseline $K$-means initializers \kmeanspp{} and \uni{}, in terms of MSEs. The data is described in Figure~\ref{fig::gridinits}, where sample initializations are also presented. Results in Figure~\ref{three_methods_time} show that \clarans{} provides significantly lower MSEs than \medlloyd{}, an observation which generalizes across data types (sequence, sparse, etc), metrics (Levenshtein, $l_\infty$, etc), and potentials (exponential, logarithmic, etc), as shown in Appendix~\ref{app::generalized}.

\begin{figure}
\begin{minipage}[c]{0.64\textwidth}
\includegraphics[width=1.0\columnwidth]{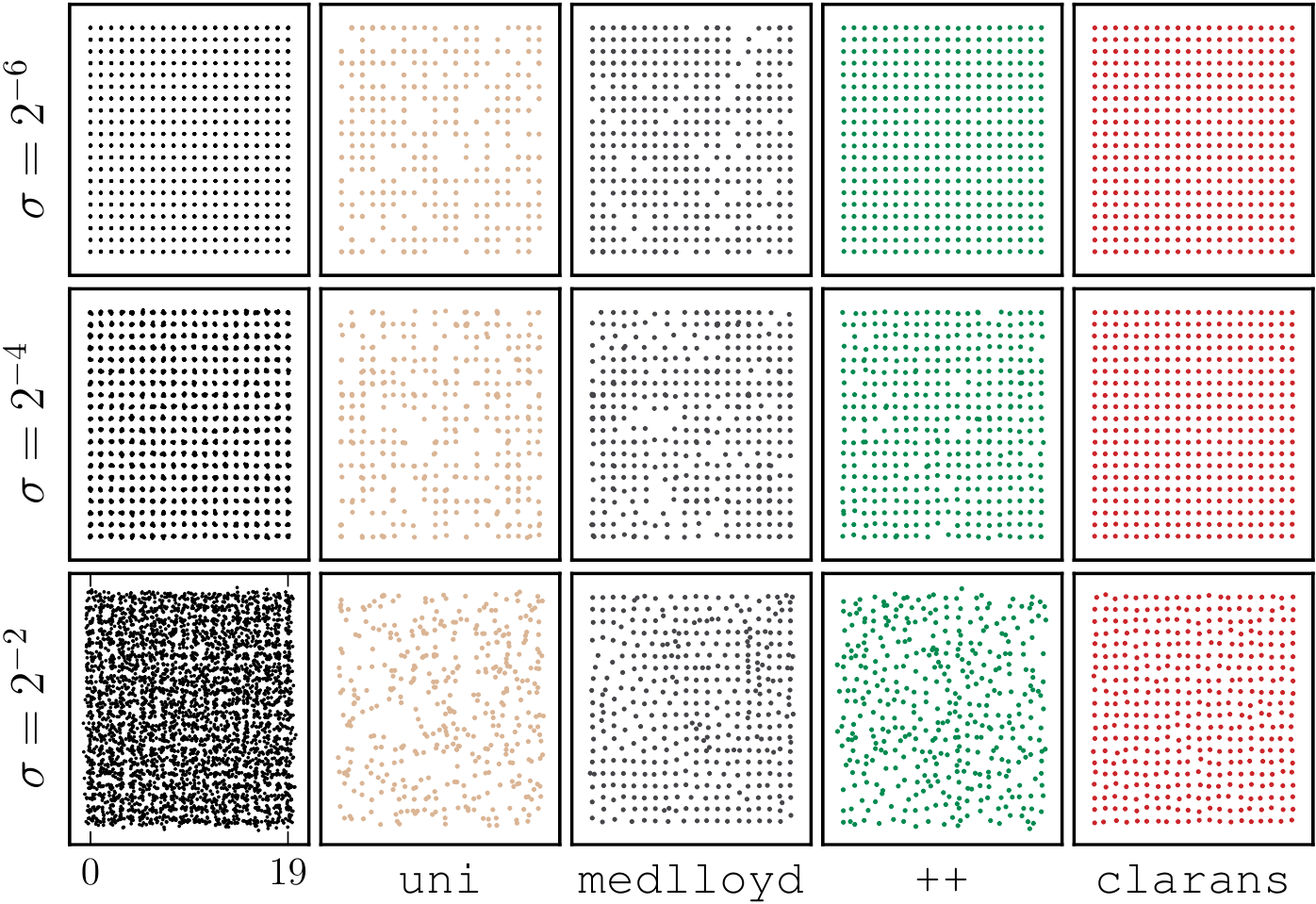}
\end{minipage}\hfill
\begin{minipage}[c]{0.34\textwidth}
\caption{(\textit{Column 1}) Simulated data in $\mathbf{R}^2$. For each cluster center $g \in \{0, \ldots, 19 \}^2 $, 100 points are drawn from $\mathcal{N}(g, \sigma^2 I)$, illustrated here for $\sigma \in \{2^{-6}, 2^{-4}, 2^{-2}\}$. (\textit{Columns 2,3,4,5}) Sample initializations. We observe `holes' for methods \texttt{uni}, \texttt{medlloyd} and \kmeanspp{}. \clarans{} successfully fills holes by removing distant, under-utilised centers. The spatial correlation of \medlloyd{}'s holes are due to its locality of updating.  }
\label{fig::gridinits}
\end{minipage}
\end{figure}

\begin{figure}
\begin{center}
\includegraphics[width=1.0\columnwidth]{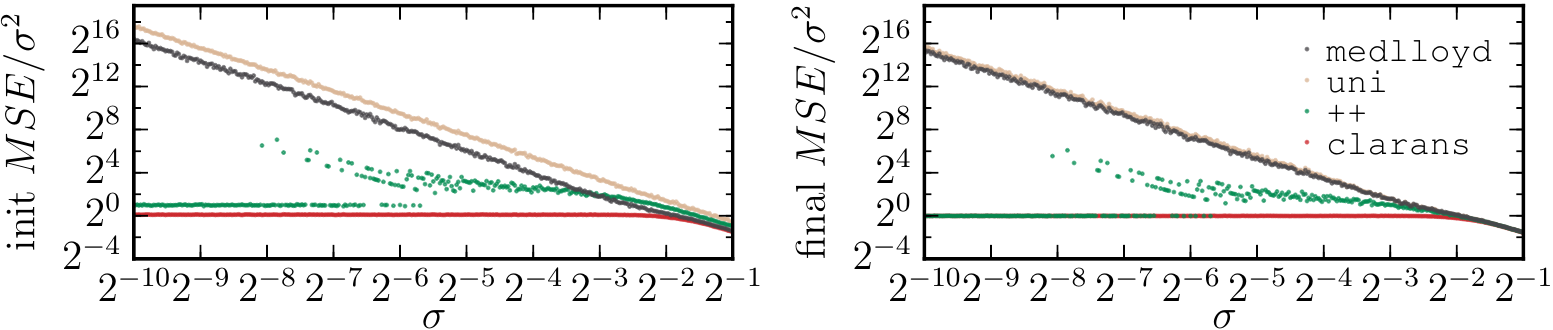}
\caption{Results on simulated data. For 400 values of $\sigma \in [2^{-10}, 2^{-1}]$, initialization (left) and final (right) MSEs relative to true cluster variances. For $\sigma \in [2^{-5}, 2^{-2}]$ \kmeanspp{} never results in minimal MSE ($MSE/\sigma^2 = 1$), while \clarans{} does for all $\sigma$. Initialization MSE with \medlloyd{} is on average 4 times lower than with \uni{}, but most of this improvement is regained when \lloyd{} is subsequently run (final $MSE/\sigma^2$).}
\label{three_methods_time}
\end{center}
\end{figure}

\section{Complexity and Accelerations}
\label{sec::complexity}
\lloyd{} requires $KN$ distance calculations to update $K$ centers, assuming no acceleration technique such as that of~\citet{elkan_2003_kmeansicml} is used. The cost of several iterations of \lloyd{} outweighs initialization with any of \uni{}, \kmeanspp{} and \mcmc{}. We ask if the same is true with \clarans{} initialization, and find that the answer depends on how \clarans{} is implemented. \clarans{} as presented in~\citet{clarans} is $O(N^2)$ in computation and memory, making it unusable for large datasets. To make \clarans{} scalable, we have investigated ways of implementing it in $O(N)$ memory, and devised optimisations which make its complexity equivalent to that of \lloyd{}. 


\clarans{} consists of two main steps. The first is swap \textit{evaluation} (line 6) and the second is swap \textit{implementation} (scope of if-statement at line 8). Proposing a good swap becomes less probable as MSE decreases, thus as the number of swap implementations increases the number of consecutive rejected proposals $(n_r)$ is likely to grow large, illustrated in Figure~\ref{eval_vs_impl}. This results in a larger fraction of time being spent in the evaluation step.  

We will now discuss optimisations in order of increasing algorithmic complexity, presenting their computational complexities in terms of evaluation and implementation steps. The explanations here are high level, with algorithmic details and pseudocode deferred to \S\ref{app::clarans}.

\begin{figure}
\begin{center}
\includegraphics[width=1.0\columnwidth]{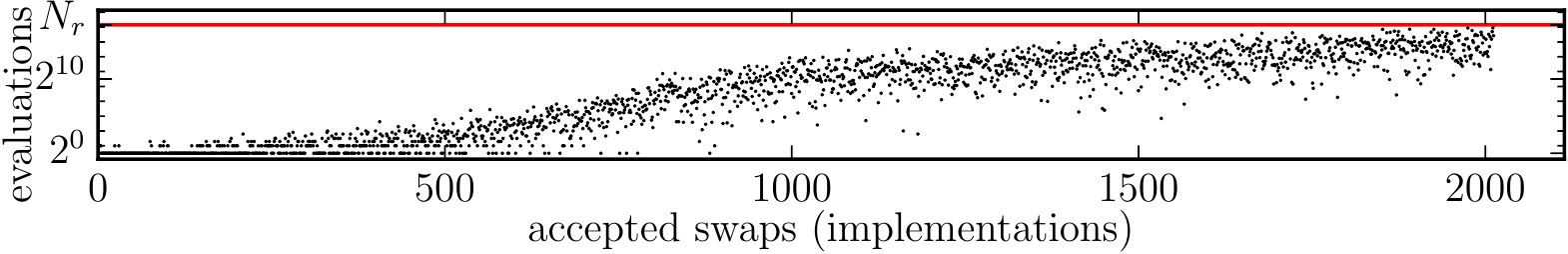}
\caption{The number of consecutive swap proposal rejections (evaluations) before one is accepted (implementations), for simulated data (\S\ref{sec::simstudy}) with $\sigma = 2^{-4}$.}
\label{eval_vs_impl}
\end{center}
\vspace{-1em}
\end{figure}

\textbf{Level -2}\;\; To evaluate swaps (line 6), simply compute all $KN$ distances.

\textbf{Level -1}\;\; Keep track of nearest centers. Now to evaluate a swap, samples whose nearest center is $x(i_-)$ need distances to all $K$ samples indexed by $\mathcal{C} \setminus \{i_-\} \cup \{i_+\}$ computed in order to determine the new nearest. Samples whose nearest is not $x(i_-)$ only need the distance to $x(i_+)$ computed to determine their nearest, as either, (1) their nearest is unchanged, or (2) it is $x(i_+)$. 

\textbf{Level 0}\;\; Also keep track of second nearest centers, as in the implementation of \citet{clarans}, which recall is $O(N^2)$ in memory and computes all distances upfront. Doing so, nearest centers can be determined for \textit{all} samples by computing distances to $x(i_+)$. If swap $(i_-, i_+)$ is accepted, samples whose new nearest is $x(i_+)$ require $K$ distance calculations to recompute second nearests. Thus from level -1 to 0, computation is transferred from evaluation to implementation, which is good, as implementation is less frequently performed, as illustrated in Figure~\ref{eval_vs_impl}. 

\textbf{Level 1}\;\; Also keep track, for each cluster center, of the distance to the furthest cluster member as well as the maximum, over all cluster members, of the minimum distance to another center. Using the triangle inequality, one can then frequently eliminate computation for clusters which are unchanged by proposed swaps with just a single center-to-center distance calculation. Note that using the triangle inequality requires that the $K$-medoids dissimilarity is metric based, as is the case in the $K$-means initialization setting.

\textbf{Level 2}\;\; Also keep track of center-to-center distances. This allows whole clusters to be tagged as unchanged by a swap, without computing any distances in the evaluation step.

We have also considered optimisations which, unlike levels -2 to 2, do not result in the exact same clustering as \clarans{}, but provide additional  acceleration. One such optimisation uses random sub-sampling to evaluate proposals, which helps significantly when $N/K$ is large. Another optimisation which is effective during initial rounds is to \textit{not} implement the first MSE reducing swap found, but to rather continue searching for approximately as long as swap implementation takes, thus balancing time between searching (evaluation) and implementing swaps. Details can be found in \S\ref{sec:level3}.

The computational complexities of these optimisations are in Table~\ref{tab::complexity}. Proofs of these complexities rely on there being $O(N/K)$ samples changing their nearest or second nearest center during a swap. In other words, for any two clusters of sizes $n_1$ and $n_2$, we assume $n_1 = \Omega(n_2)$. Using level 2 complexities, we see that if a fraction $p(\mathcal{C})$ of proposals reduce MSE, then the expected complexity is $O( N(1 + 1/(p(\mathcal{C})K)))$. One cannot marginalise $\mathcal{C}$ out of the expectation, as $\mathcal{C}$ may have no MSE reducing swaps, that is $p(\mathcal{C}) = 0$. If $p(\mathcal{C})$ is $O(K)$, we obtain complexity $O(N)$ per swap, which is equivalent to the $O(KN)$ for $K$ center updates of \lloyd{}. In Table~\ref{tab::distcalcs}, we consider run times and distance calculation counts on simulated data at the various levels of optimisation.

\begin{table}
\addtolength{\tabcolsep}{+2pt}  
\begin{tabular}{cccccc}
\hline
 & -2 & -1 & 0 & 1 & 2 \\ 
\hline
1 evaluation & $NK$ & $N$ & $N$ & $\frac{N}{K} + K$ & $\frac{N}{K}$ \\ 
1 implementation & $1$ & $1$ & $N$ & $N$ & $N$ \\ 
$K^2$ evaluations, $K$ implementations & $K^3N$ & $K^2N$ & $K^2N$ & $NK + K^3$ & $KN$ \\ 
memory & $N$ & $N$ & $N$ & $N$ & $N + K^2$ \\ 

\hline
\end{tabular}
\addtolength{\tabcolsep}{-2pt}  
\vspace{1em}
\caption{The complexities at different levels of optimisation of \textit{evaluation} and \textit{implementation}, in terms of required distance calculations, and overall memory. We see at level 2 that to perform $K^2$ evaluations and $K$ implementations is $O(KN)$, equivalent to \lloyd{}.} 
\vspace{-1em} 
\label{tab::complexity}
\end{table}

\begin{table}
\begin{minipage}[c]{0.4\textwidth}
\addtolength{\tabcolsep}{-2pt}  
\begin{tabular}{cccccc}
\hline
 & -2 & -1 & 0 & 1 & 2 \\ 
\hline
$\log_2 ( \#\mbox{\,dcs\,})$ & $44.1$ & $36.5$ & $35.5$ & $29.4$ & $26.7$ \\ 
time [s] & - & - & 407 & 19.2 & 15.6 \\ 
\hline
\end{tabular}
\addtolength{\tabcolsep}{+2pt}  
\end{minipage}\hfill
\begin{minipage}[c]{0.50\textwidth}
\vspace{1em}
\caption{Total number of distance calculations ($\#\mbox{\,dcs\,}$) and time required by \clarans{} on simulation data of \S\ref{sec::simstudy} with $\sigma = 2^{-4}$  at different optimisation levels. } 
 \label{tab::distcalcs}
\end{minipage}
\end{table}

\section{Results}
\label{sec::mainresults}

\begin{table}[h!]
\addtolength{\tabcolsep}{-1.2pt}  
\begin{tabular}{cccccc}
\hline
dataset & $\#$ & N & dim & K & TL $[s]$\\
\hline
 a1 & 1 & 3000 & 2 & 40 & 1.94 \\
 a2 & 2 & 5250 & 2 & 70 & 1.37 \\
 a3 & 3 & 7500 & 2 & 100 & 1.69 \\
 birch1 & 4 & 100000 & 2 & 200 & 21.13 \\
 birch2 & 5 & 100000 & 2 & 200 & 15.29 \\
 birch3 & 6 & 100000 & 2 & 200 & 16.38 \\
 ConfLong & 7 & 164860 & 3 & 22 & 30.74 \\
 dim032 & 8 & 1024 & 32 & 32 & 1.13 \\
 dim064 & 9 & 1024 & 64 & 32 & 1.19 \\
 dim1024 & 10 & 1024 & 1024 & 32 & 7.68 \\
 europe & 11 & 169308 & 2 & 1000 & 166.08 \\
\hline
\end{tabular}
\begin{tabular}{cccccc}
\hline
dataset & $\#$ & N & dim & K & TL $[s]$\\
\hline
 housec8 & 12 & 34112 & 3 & 400 & 18.71 \\
 KDD${^*}$ & 13 & 145751 & 74 & 200 & 998.83 \\
 mnist & 14 & 10000 & 784 & 300 & 233.48 \\
 Mopsi & 15 & 13467 & 2 & 100 & 2.14 \\
 rna${^*}$ & 16 & 20000 & 8 & 200 & 6.84 \\
 s1 & 17 & 5000 & 2 & 30 & 1.20 \\
 s2 & 18 & 5000 & 2 & 30 & 1.50 \\
 s3 & 19 & 5000 & 2 & 30 & 1.39 \\
 s4 & 20 & 5000 & 2 & 30 & 1.44 \\
 song${^*}$ & 21 & 20000 & 90 & 200 & 71.10 \\
 susy${^*}$ & 22 & 20000 & 18 & 200 & 24.50 \\
 yeast & 23 & 1484 & 8 & 40 & 1.23 \\
\hline
\end{tabular}

\addtolength{\tabcolsep}{+1.2pt}
\vspace{1em}
\caption{The 23 datasets. Column `TL' is time allocated to run with each initialization scheme, so that no new runs start after TL elapsed seconds. The starred datasets are those used in~\citet{bachem16tortoise}, the remainder are available at \protect\url{https://cs.joensuu.fi/sipu/datasets}.}
\label{tab:datasets}
\end{table}

We first compare \clarans{} with \uni{}, \kmeanspp{}, \mcmc{} and \BF{} on the first 23 publicly available datasets in Table~\ref{tab:datasets} (datasets 1-23).  
As noted in \citet{celebi_2013_initialisation}, it is common practice to run initialization+\lloyd{} several time and retain the solution with the lowest MSE. In~\citet{bachem16tortoise} methods are run a fixed number of times, and \textit{mean} MSEs are compared. However, when comparing \textit{minimum} MSEs over several runs, one should take into account that methods vary in their time requirements.

Rather than run each method a fixed number of times, we therefore run each method as many times as possible in a given time limit, `TL'. This dataset dependent time limit, given by columns TL in Table~\ref{tab:datasets}, is taken as 80$\times$ the time of a single run of \kmeanspp{}+\lloyd{}. The numbers of runs completed in time TL by each method are in columns 1-5 of Table~\ref{tab::mega}. Recall that our stopping criterion for \clarans{} is $K^2$ consecutively rejected swap proposals. We have also experimented with stopping criterion based on run time and number of swaps implemented, but find that stopping based on number of rejected swaps best guarantees convergence. We use $K^2$ rejections for simplicity, although have found that fewer than $K^2$ are in general needed to obtain minimal MSEs.

We use the fast \lloyd{} implementation accompanying~\citet{fast_bounds} with the `auto' flag set to select the best exact accelerated algorithm, and run until complete convergence. For initializations, we use our own C++/Cython implementation of level 2 optimised \clarans{}, the implementation of \mcmc{} of \citet{bachem16tortoise}, and \kmeanspp{} and \BF{} of~\citet{fast_bounds}.

The objective of~\citet{bachem16tortoise} was to prove and experimentally validate that~\mcmc{} produces initialization MSEs equivalent to those of \kmeanspp{}, and as such \lloyd{} was not run during experiments. We consider both initialization MSE, as in \citet{bachem16tortoise}, and final MSE after \lloyd{} has run. The latter is particularly important, as it is the objective we wish to minimize in the $K$-means problem.

In addition to considering initialization and final MSEs, we also distinguish between mean and minimum MSEs. We believe the latter is important as it captures the varying time requirements, and as mentioned it is common to run \lloyd{} several times and retain the lowest MSE clustering. In Table~\ref{tab::mega} we consider two MSEs, namely mean initialization MSE and minimum final MSE.

\begin{table}
\centering
\addtolength{\tabcolsep}{-1.01pt}  

\begin{tabular}{c|ccccc|cccc|ccccc}

&  \multicolumn{5} {c|} {runs completed} & \multicolumn{4} {c|} {mean initial mse} &  \multicolumn{5} {c} {minimum final mse} \\
& \begin{sideways} \texttt{km}  \end{sideways} \begin{sideways} \texttt{++}  \end{sideways} & \begin{sideways} \texttt{afk} \end{sideways} \begin{sideways} \texttt{mc2} \end{sideways} & \begin{sideways} \texttt{uni}   \end{sideways} & \begin{sideways} \texttt{bf}    \end{sideways} & \begin{sideways} \texttt{cla}  \end{sideways} \begin{sideways} \texttt{rans} \end{sideways} & \begin{sideways} \texttt{km}  \end{sideways} \begin{sideways} \texttt{++}  \end{sideways} & \begin{sideways} \texttt{afk} \end{sideways} \begin{sideways} \texttt{mc2} \end{sideways} & \begin{sideways} \texttt{uni}   \end{sideways} & \begin{sideways} \texttt{cla}  \end{sideways} \begin{sideways} \texttt{rans} \end{sideways} & \begin{sideways} \texttt{km}  \end{sideways} \begin{sideways} \texttt{++}  \end{sideways} & \begin{sideways} \texttt{afk} \end{sideways} \begin{sideways} \texttt{mc2} \end{sideways} & \begin{sideways} \texttt{uni}   \end{sideways} & \begin{sideways} \texttt{bf}    \end{sideways} & \begin{sideways} \texttt{cla}  \end{sideways} \begin{sideways} \texttt{rans} \end{sideways} \\
\hline 
1  &  135 & 65 & 138 & 8 & 29 & 1 & 0.97 & 2 & \textbf{0.63}  & 0.59 & 0.58 & 0.59 & 0.61 & \textbf{0.57} \\ 
2  &  81 & 24 & 85 & 5 & 7 & 1 & 0.99 & 1.96 & \textbf{0.62}  & 0.6 & 0.59 & 0.61 & 0.63 & \textbf{0.58} \\ 
3  &  82 & 21 & 87 & 6 & 4 & 1 & 0.99 & 2.07 & \textbf{0.63}  & 0.6 & 0.61 & 0.62 & 0.63 & \textbf{0.59} \\ 
4  &  79 & 27 & 95 & 28 & 5 & 1 & 0.99 & 1.54 & \textbf{0.69} & 0.66 & 0.66 & 0.66 & 0.66 & \textbf{0.66} \\ 
5  &  85 & 22 & 137 & 27 & 6 & 1 & 1 & 3.8 & \textbf{0.62} & 0.62 & 0.62 & 0.64 & 0.63 & \textbf{0.59} \\ 
6  &  68 & 22 & 77 & 23 & 4 & 1 & 0.98 & 2.35 & \textbf{0.67} & 0.64 & 0.64 & 0.68 & 0.68 & \textbf{0.63} \\ 
7  &  84 & 66 & 75 & 38 & 46 & 1 & 1 & 1.17 & \textbf{0.73}  & 0.64 & 0.64 & 0.64 & \textbf{0.64} & 0.64 \\ 
8  &  84 & 29 & 88 & 5 & 19 & 1 & 0.98 & 43.1 & \textbf{0.65} & 0.65 & 0.65 & 0.66 & 0.66 & \textbf{0.63} \\ 
9  &  81 & 29 & 90 & 5 & 16 & 1 & 1.01 & $>\!\!\!10^2$ & \textbf{0.66}  & 0.66 & 0.66 & 0.66 & 0.69 & \textbf{0.63} \\ 
10  &  144 & 52 & 311 & 24 & 18 & 1 & 0.99 & $>\!\!\!10^2$ & \textbf{0.72} & 0.62 & 0.61 & 0.62 & 0.62 & \textbf{0.59} \\ 
11  &  70 & 25 & 28 & 15 & 4 & 1 & 1 & 20.2 & \textbf{0.72} & 0.67 & 0.67 & 2.25 & 2.4 & \textbf{0.64} \\ 
12  &  80 & 27 & 81 & 21 & 4 & 1 & 0.99 & 2.09 & \textbf{0.77} & 0.7 & 0.7 & 0.73 & 0.74 & \textbf{0.69} \\ 
13  &  102 & 74 & 65 & 56 & 5 & 1 & 1 & 4 & \textbf{0.77}  & 0.69 & 0.69 & 0.75 & 0.75 & \textbf{0.69} \\ 
14  &  88 & 43 & 276 & 83 & 4 & 1 & 1 & 1 & \textbf{0.87}  & 0.6 & 0.6 & 0.6 & 0.61 & \textbf{0.6} \\ 
15  &  91 & 23 & 52 & 7 & 4 & 1 & 1 & 25 & \textbf{0.6}  & 0.57 & 0.57 & 3.71 & 3.62 & \textbf{0.51} \\ 
16  &  107 & 28 & 86 & 28 & 4 & 1 & 0.99 & 24.5 & \textbf{0.62}  & 0.62 & 0.61 & 2.18 & 2.42 & \textbf{0.56} \\ 
17  &  84 & 31 & 85 & 5 & 25 & 1 & 1.01 & 2.79 & \textbf{0.7}  & 0.66 & 0.65 & 0.67 & 0.69 & \textbf{0.65} \\ 
18  &  100 & 39 & 100 & 7 & 30 & 1 & 0.99 & 2.24 & \textbf{0.69} & 0.65 & 0.65 & 0.66 & 0.66 & \textbf{0.64} \\ 
19  &  88 & 36 & 83 & 6 & 24 & 1 & 1.05 & 1.55 & \textbf{0.71} & 0.65 & 0.65 & 0.66 & 0.67 & \textbf{0.65} \\ 
20  &  88 & 36 & 87 & 6 & 24 & 1 & 1.01 & 1.65 & \textbf{0.71} & 0.65 & 0.64 & 0.64 & 0.65 & \textbf{0.64} \\ 
21  &  96 & 52 & 98 & 67 & 4 & 1 & 1 & 1.14 & \textbf{0.8}  & 0.67 & 0.66 & 0.71 & 0.7 & \textbf{0.65} \\ 
22  &  116 & 48 & 134 & 67 & 4 & 1 & 1 & 1.04 & \textbf{0.81} & 0.69 & 0.69 & 0.69 & 0.69 & \textbf{0.69} \\ 
23  &  82 & 31 & 81 & 5 & 6 & 1 & 1 & 1.18 & \textbf{0.74}  & 0.65 & 0.65 & 0.65 & 0.67 & \textbf{0.64} \\ 
\hline
gm & 90 & 34 & 93 & 14 & 8 & 1 & 1 & 4.71 & \textbf{0.7}  & 0.64 & 0.64 & 0.79 & 0.8 & \textbf{0.62} \\ 
\hline
\end{tabular}

\addtolength{\tabcolsep}{+1.01pt}  
\vspace{1em}
\caption{Summary of results on the 23 datasets (rows). Columns 1 to 5 contain the number of initialization+\lloyd{} runs completed in time limit TL. Columns 6 to 14 contain MSEs relative to the mean initialization MSE of \kmeanspp{}. Columns 6 to 9 are mean MSEs after initialization but before \lloyd{}, and columns 10 to 14 are minimum MSEs after \lloyd{}. The final row (gm) contains geometric means of all columns.  \clarans{} consistently obtains the lowest across all MSE measurements, and has a 30\% lower initialization MSE than \kmeanspp{} and \mcmc{}, and a 3\% lower final minimum MSE. }
\label{tab::mega}
\end{table}

\begin{figure}
\begin{center}
\includegraphics[width=1.0\columnwidth]{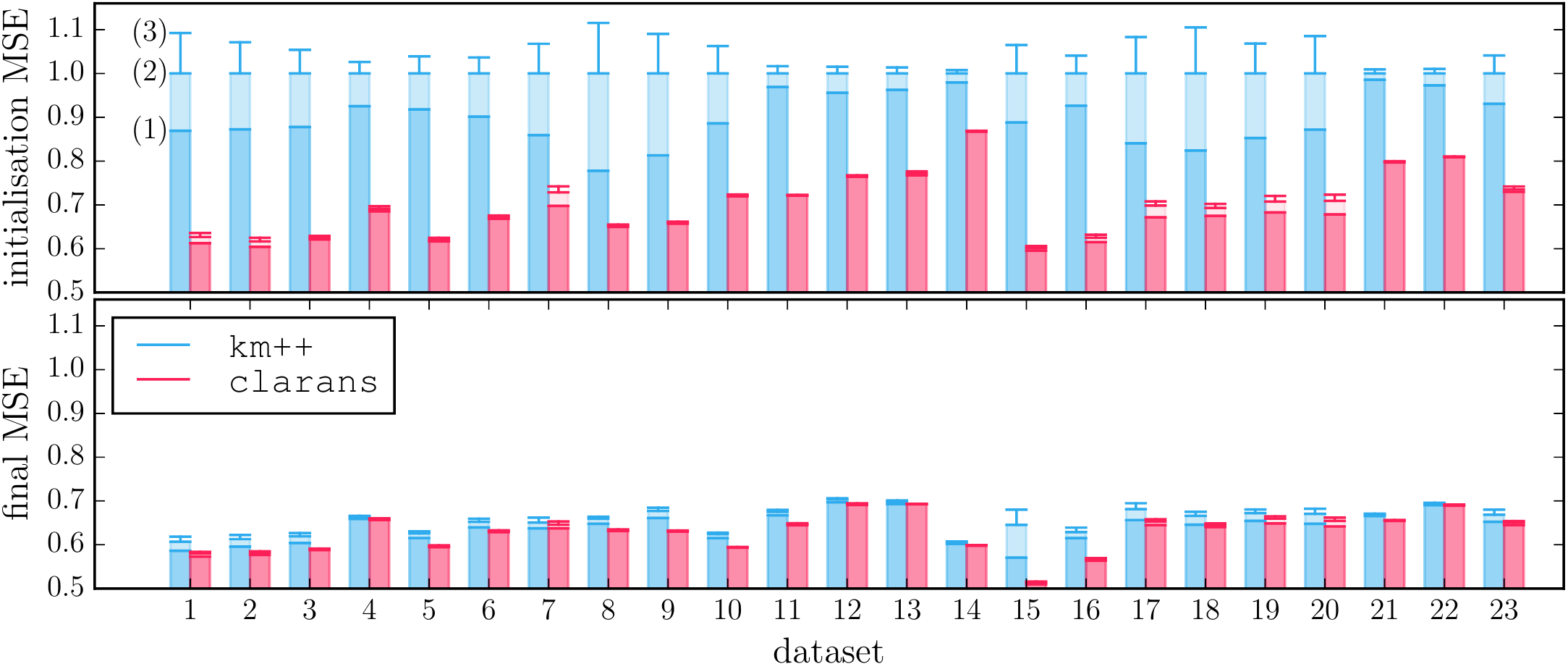}
\caption{Initialization (above) and final (below) MSEs for \kmeanspp{} (left bars) and \clarans{} (right bars), with minumum (1), mean (2) and mean + standard deviation (3) of MSE across all runs. For all initialization MSEs and most final MSEs, the lowest \kmeanspp{} MSE is several standard deviations higher than the mean \clarans{} MSE.   } 
\label{bothMSEs}
\end{center}
\end{figure}

\subsection{Baseline performance}

We briefly discuss findings related to algorithms \uni{}, \BF{}, \mcmc{} and \kmeanspp{}. Results in Table~\ref{tab::mega} corroborate the previously established finding that \uni{} is vastly outperformed by \kmeanspp{}, both in initialization and final MSEs. Table~\ref{tab::mega} results also agree with the finding of \citet{bachem16tortoise} that initialization MSEs with~\mcmc{} are indistinguishable from those of \kmeanspp{}, and moreover that final MSEs are indistinguishable. We observe in our experiments that runs with \kmeanspp{} are faster than those with \mcmc{} (columns 1 and 2 of Table~\ref{tab::mega}). We attribute this to the fast blas-based \kmeanspp{} implementation of~\citet{fast_bounds}. 

Our final baseline finding is that MSEs obtained with \BF{} are in general no better than those with \uni{}. This is not in strict agreement with the findings of \citet{celebi_2013_initialisation}. We attribute this discrepancy to the fact that experiments in \citet{celebi_2013_initialisation} are in the low $K$ regime $(K <  50, \, N/K > 100)$. Note that Table~\ref{tab::mega} does not contain initialization MSEs for \BF{}, as \BF{} does not initialize with data points but with means of sub-samples, and it would thus not make sense to compare \BF{} initialization with the 4 seeding methods.

\subsection{\clarans{} performance}
Having established that the best baselines are \kmeanspp{} and \mcmc{}, and that they provide clusterings of indistinguishable quality, we now focus on the central comparison of this paper, that between \kmeanspp{} with \clarans{}. In Figure~\ref{bothMSEs} we present bar plots summarising all runs on all 23 datasets. We observe a very low variance in the initialization MSEs of \clarans{}. We speculatively hypothesize that \clarans{} often finds a globally minimal initialization.  Figure~\ref{bothMSEs} shows that \clarans{} provides significantly lower initialization MSEs than \kmeanspp{}.

The final MSEs are also significantly better when initialization is done with \clarans{}, although the gap in MSE between \clarans{} and \kmeanspp{} is reduced when \lloyd{} has run. Note, as seen in Table~\ref{tab::mega}, that all 5 initializations for dataset 7 result in equally good clusterings.    

As a supplementary experiment, we considered initialising with \kmeanspp{} and \clarans{} in series, thus using the three stage clustering \kmeanspp{}+\clarans{}+\lloyd{}. We find that this can be slightly faster than just \clarans{}+\lloyd{} with identical MSEs. Results of this experiment are presented in \S\ref{sec::precede2}. We perform a final experiment measure the dependence of improvement on $K$ in \S\ref{sec::precede2}, where we see the improvement is most significant for large $K$.  

\section{Conclusion and Future Works}
In this paper, we have demonstrated the effectiveness of the algorithm \clarans{} at solving the $k$-medoids problem. We have described techniques for accelerating \clarans{}, and most importantly shown that \clarans{} works very effectively as an initializer for \lloyd{}, outperforming other initialization schemes, such as \kmeanspp{}, on 23 datasets.

An interesting direction for future work might be to develop further optimisations for \clarans{}. One idea could be to use importance sampling to rapidly obtain good estimates of post-swap energies. Another might be to propose two swaps simultaneously, as considered in~\cite{Kanungo_2002_local}, which could potentially lead to even better solutions, although we have hypothesized that \clarans{} is already finding globally optimal initializations. 

All source code is made available under a public license. It consists of generic C++ code which can be extended to various data types and metrics, compiling to a shared library with extensions in Cython for a Python interface. It can currently be found in the git repository \url{https://github.com/idiap/zentas}.

\section*{Acknowledgments}
James Newling was funded by the Hasler Foundation under the grant 13018 MASH2.

\newpage

\bibliographystyle{hapalike}

\bibliography{initclarans} 

\begin{thebibliography}{}

\bibitem[Arthur and Vassilvitskii, 2007]{arthur_2007_kmeanspp}
Arthur, D. and Vassilvitskii, S. (2007).
\newblock K-means++: The advantages of careful seeding.
\newblock In {\em Proceedings of the Eighteenth Annual ACM-SIAM Symposium on
  Discrete Algorithms}, SODA '07, pages 1027--1035, Philadelphia, PA, USA.
  Society for Industrial and Applied Mathematics.

\bibitem[Bachem et~al., 2016]{bachem16tortoise}
Bachem, O., Lucic, M., Hassani, S.~H., and Krause, A. (2016).
\newblock Fast and provably good seedings for k-means.
\newblock In {\em Neural Information Processing Systems (NIPS)}.

\bibitem[Bradley and Fayyad, 1998]{bradley_fayyad}
Bradley, P.~S. and Fayyad, U.~M. (1998).
\newblock Refining initial points for k-means clustering.
\newblock In {\em Proceedings of the Fifteenth International Conference on
  Machine Learning}, ICML '98, pages 91--99, San Francisco, CA, USA. Morgan
  Kaufmann Publishers Inc.

\bibitem[Celebi et~al., 2013]{celebi_2013_initialisation}
Celebi, M.~E., Kingravi, H.~A., and Vela, P.~A. (2013).
\newblock A comparative study of efficient initialization methods for the
  k-means clustering algorithm.
\newblock {\em Expert Syst. Appl.}, 40(1):200--210.

\bibitem[Elkan, 2003]{elkan_2003_kmeansicml}
Elkan, C. (2003).
\newblock Using the triangle inequality to accelerate k-means.
\newblock In {\em Machine Learning, Proceedings of the Twentieth International
  Conference {(ICML} 2003), August 21-24, 2003, Washington, DC, {USA}}, pages
  147--153.

\bibitem[Hartigan, 1975]{Hartigan_1975_book}
Hartigan, J.~A. (1975).
\newblock {\em Clustering Algorithms}.
\newblock John Wiley \& Sons, Inc., New York, NY, USA, 99th edition.

\bibitem[Hastie et~al., 2001]{the_elements}
Hastie, T.~J., Tibshirani, R.~J., and Friedman, J.~H. (2001).
\newblock {\em The elements of statistical learning : data mining, inference,
  and prediction}.
\newblock Springer series in statistics. Springer, New York.

\bibitem[Kanungo et~al., 2002]{Kanungo_2002_local}
Kanungo, T., Mount, D.~M., Netanyahu, N.~S., Piatko, C.~D., Silverman, R., and
  Wu, A.~Y. (2002).
\newblock A local search approximation algorithm for k-means clustering.
\newblock In {\em Proceedings of the Eighteenth Annual Symposium on
  Computational Geometry}, SCG '02, pages 10--18, New York, NY, USA. ACM.

\bibitem[Kaufman and Rousseeuw, 1990]{generic_kmedoids}
Kaufman, L. and Rousseeuw, P.~J. (1990).
\newblock {\em Finding groups in data : an introduction to cluster analysis}.
\newblock Wiley series in probability and mathematical statistics. Wiley, New
  York.
\newblock A Wiley-Interscience publication.

\bibitem[Lewis et~al., 2004]{Lewis_2004_rcv1}
Lewis, D.~D., Yang, Y., Rose, T.~G., and Li, F. (2004).
\newblock Rcv1: A new benchmark collection for text categorization research.
\newblock {\em Journal of Machine Learning Research}, 5:361--397.

\bibitem[Newling and Fleuret, 2016]{fast_bounds}
Newling, J. and Fleuret, F. (2016).
\newblock Fast k-means with accurate bounds.
\newblock In {\em Proceedings of the International Conference on Machine
  Learning (ICML)}, pages 936--944.

\bibitem[Ng and Han, 1994]{clarans}
Ng, R.~T. and Han, J. (1994).
\newblock Efficient and effective clustering methods for spatial data mining.
\newblock In {\em Proceedings of the 20th International Conference on Very
  Large Data Bases}, VLDB '94, pages 144--155, San Francisco, CA, USA. Morgan
  Kaufmann Publishers Inc.

\bibitem[Ng and Han, 2002]{clarans_other}
Ng, R.~T. and Han, J. (2002).
\newblock Clarans: A method for clustering objects for spatial data mining.
\newblock {\em IEEE Transactions on Knowledge and Data Engineering}, pages
  1003--1017.

\bibitem[Park and Jun, 2009]{park_2009_kmedoids}
Park, H.-S. and Jun, C.-H. (2009).
\newblock A simple and fast algorithm for k-medoids clustering.
\newblock {\em Expert Syst. Appl.}, 36(2):3336--3341.

\bibitem[Telgarsky and Vattani, 2010]{telgarsky_2010_hartigan}
Telgarsky, M. and Vattani, A. (2010).
\newblock Hartigan's method: k-means clustering without voronoi.
\newblock In {\em {AISTATS}}, volume~9 of {\em {JMLR} Proceedings}, pages
  820--827. JMLR.org.

\bibitem[Wu et~al., 2008]{wu2008wu}
Wu, X., Kumar, V., Quinlan, J.~R., Ghosh, J., Yang, Q., Motoda, H., McLachlan,
  G., Ng, A., Liu, B., Yu, P., Zhou, Z.-H., Steinbach, M., Hand, D., and
  Steinberg, D. (2008).
\newblock Top 10 algorithms in data mining.
\newblock {\em Knowledge and Information Systems}, 14(1):1--37.

\bibitem[Yujian and Bo, 2007]{normalised_Levenshtein}
Yujian, L. and Bo, L. (2007).
\newblock A normalized levenshtein distance metric.
\newblock {\em IEEE Trans. Pattern Anal. Mach. Intell.}, 29(6):1091--1095.

\end{thebibliography}

\clearpage

\renewcommand{\thesection}{SM-\Alph{section}}

\setcounter{section}{0}

\onecolumn

\section{Generalised $K$-Medoids Results}
\label{app::generalized}

The potential uses of \clarans{} as a $K$-medoids algorithm go well beyond $K$-Means initialization. In this Appendix, we wish to demonstrate that \clarans{} should be chosen as a default $K$-medoids algorithm, rather than \medlloyd{}. In its most general form, the $K$-medoids problem is to minimize,
\begin{equation}
\label{eq:E}
E(\mathcal{C}) = \frac{1}{N}\sum_{i = 1}^{N} \argmin_{i' \in \mathcal{C}} f(x(i),x(i')).
\end{equation}
We assume that $f$ is of the form,
\begin{equation}
\label{fdecomp}
f(x(i), x(i')) = \energy(\dist(x(i),x(i'))),
\end{equation}
where $\energy$ is non-decreasing, and samples belong to a metric space with metric $\dist(\cdot, \cdot)$. Constraint~\ref{fdecomp} allows us to use the triangle inequality to eliminate certain distance calculations. We now present examples comparing \clarans{} and \medlloyd{} in various settings, showing the effectiveness of \clarans{}. Table~\ref{tab:syn_data} describes artificial problems, with results in Figure~\ref{vik_clarans_sim}. Table~\ref{tab:real_data_kmedoids} describes real-world problems, with results in Figure~\ref{vik_clarans_real}.

\begin{table}[h!]

\begin{center}
\begin{tabular}{|c|ccccc|}
\hline
& N & K & type & metric & $\energy(d)$ \\
\hline
syn-1 & 2000 & 40 & sequence & Levenshtein & $d$ \\
syn-2 & 20000 & 100 & sparse-v & $l_2$ & $d^2$ \\
syn-3 & 28800 & 144 & dense-v & $l_1$ & $e^d$ \\
syn-4 & 20000 & 100 & dense-v & $l_{\infty}$ & $I_{d > 0.05}$ \\
\hline
\end{tabular}
\end{center}

\caption{Synthetic datasets used for comparing $K$-medoids algorithms (Figure~\ref{vik_clarans_sim}).  
\textbf{syn-1}: Each of the cluster centers is a random binary sequence of 16 bits (0/1). In each of the clusters, 50 elements are generated by applying 2 mutations (insert/delete/replacement) to the center, at random locations.  
\textbf{syn-2}: Each of the centers is a vector in $\mathbb{R}^{10^6}$, non-zero at exactly 5 indices, with the 5 non-zero values drawn from $N(0,1)$. Each sample is a linear combination of two centers, with coefficients $1$ and $Q$ respectively, where $Q\sim U[-0.5,0.5]$.  
\textbf{syn-3}: Centers are integer co-ordinates of an $12\times12$ grid.  For each center, 50 samples are generated, each sample being the center plus Gaussian noise of identity covariance, as in the simulation data in the main text.  
\textbf{syn-4}: Data are points drawn uniformly from $[0,1]^2$. We attempt cover a unit square with $100$ squares of diameter $0.1$, a task with a unique lattice solution. Points not covered have energy $1$, while covered points have energy $0$.  
}

\label{tab:syn_data}
\end{table}

\begin{table}
\begin{center}
\begin{tabular}{|c|ccccc|}
\hline
& N & K & type & metric & $\energy(d)$ \\
\hline
rcv1 & 23149 & 400 & sparse-v & $l_2$ & $d^2$ \\
genome & 400000 & 1000 & sequence & n-Levensh. & $d^2$ \\
mnist & 10000 & 400 & dense-v & $l_2$ & $d^2$ \\
words &  354983 & 1000 & sequence & Levenshtein & $d^2$ \\
\hline
\end{tabular}
\end{center}

\caption{Real datasets used for comparing $K$-medoids algorithms (Figure~\ref{vik_clarans_real}), with data urls in~\ref{urls}.  
\textbf{rcv1}: The Reuters Corpus Volume I training set of \cite{Lewis_2004_rcv1}, a sparse datasets containing news article categorisation annotation. 
\textbf{genome}: Nucleotide subsequences of lengths 10,11 or 12, randomly selected from chromosome 10 of a Homo Sapiens. Note that the normalised Levenshtein metric~\citep{normalised_Levenshtein} is used.
\textbf{mnist}: The test images of the MNIST hand-written digit dataset.
\textbf{words}: A comprehensive English language word list.  
}
\label{tab:real_data_kmedoids}
\end{table}

\begin{figure}[h!]
\begin{center}
\includegraphics[width=0.6\columnwidth]{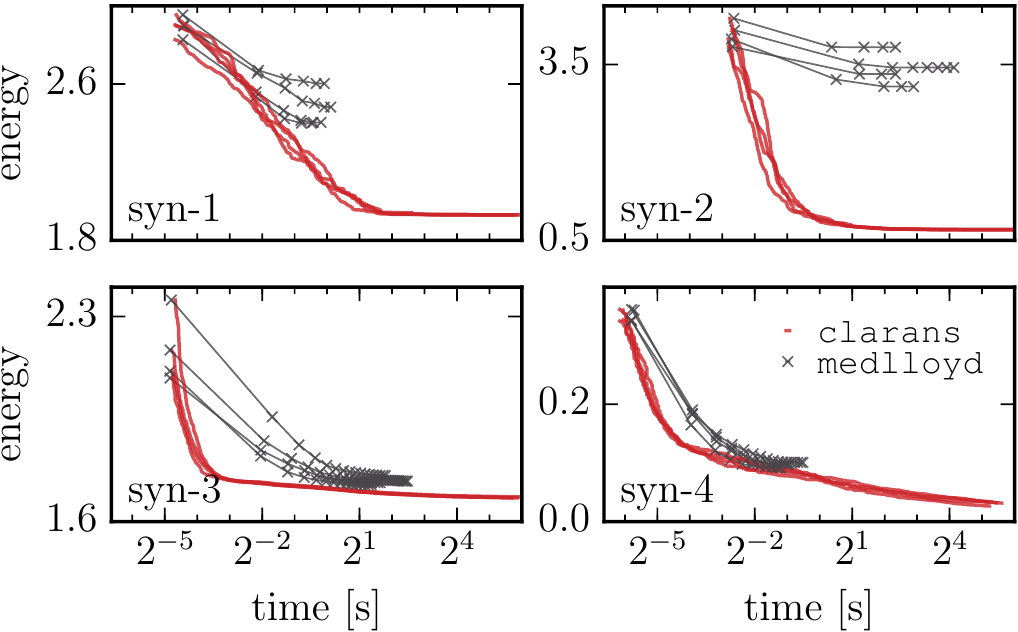}
\caption{Results on synthetic datasets. Algorithms \clarans{} and \medlloyd{} are run four times with random seedings. Each experiment is run with a time limit of 64 seconds. The vertical axis is mean energy (dissimilarity) across samples. In all experiments, \medlloyd{} gets trapped in local minima before 64 seconds have elapsed, and \clarans{} always obtains significantly lower energies than \medlloyd{}.}
\label{vik_clarans_sim}
\end{center}
\end{figure}

\begin{figure}[h!]
\begin{center}
\includegraphics[width=0.6\columnwidth]{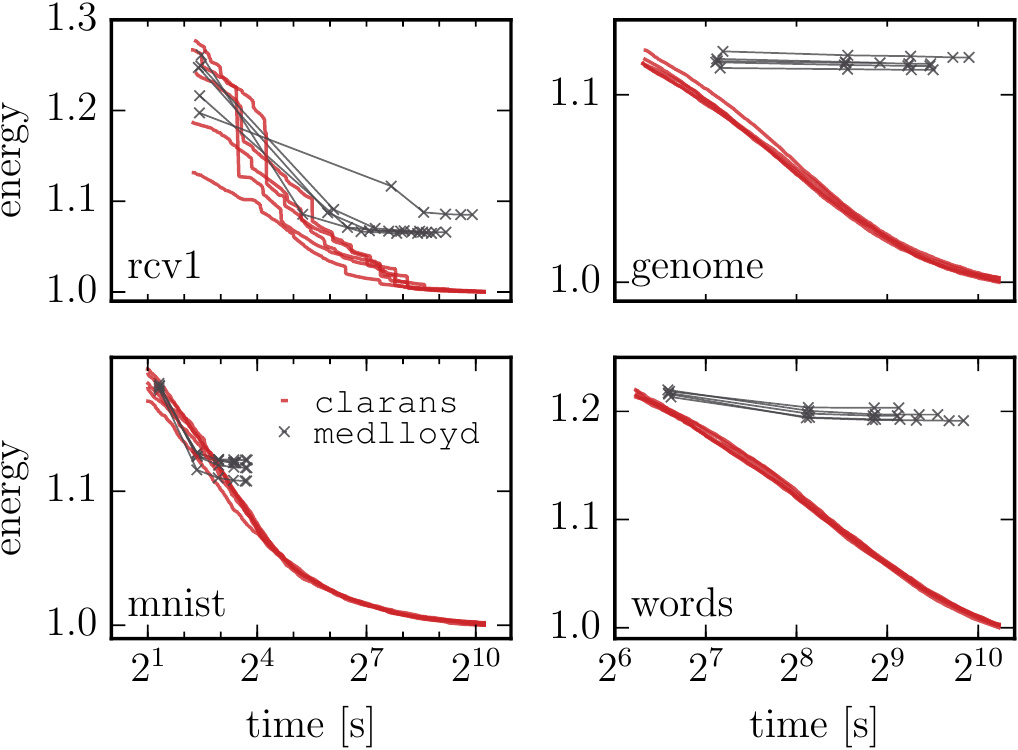}
\caption{Results on real datasets. Vertical axes are energies relative to the lowest energy found. We observe that \medlloyd{} performs very poorly on sequence datasets (right), failing to find clusterings significantly better than the random initializations. While an improvement over the initial seeding is obtained using \vik{} on the vector datasets (left), the energies obtained using \clarans{} are significantly lower. Runs with \clarans{} appear to converge to a common energy solution, even though initial energies vary greatly, as is the case in on dataset rcv1. The majority of runs with \medlloyd{} converge to a local minimum before the allotted time limit of $2^{10}$ seconds. }

\label{vik_clarans_real}
\end{center}
\vskip -0.2in
\end{figure}

\section{The task}
\label{sec:the_task}
We state precisely the $K$-medoids task in the setting where dissimilarity is an increasing function of a distance function. Given a set of $N$ elements, $\{x(i) : i \in \{1, \ldots, N\}\}$, with a distance defined between elements,
\begin{align*}
\dist(x(i), x(i')) &\ge  0, \\
\dist(x(i), x(i)) &= 0, \\
\dist(x(i), x(i')) &= \dist(x(i'), x(i)), \\
\dist(x(i), x(i'')) & \le \dist(x(i), x(i')) + \dist(x(i'), x(i'')),
\end{align*}
and given an energy function $\energy : \mathbb{R}^+ \rightarrow \mathbb{R}^+$ satisfying,
\begin{align*}
\energy(0) &= 0, \\
v_1 \le v_2 &\Longleftrightarrow \energy(v_1) \le \energy(v_2),
\end{align*}
The task is to find indices $\{ c(k) : k \in \{1, \ldots, K\} \} \subset \{1, \ldots, N\}$, to minimize,
\begin{equation*}
\sum_{i = 1}^{N} \min_{k \in \{1, \ldots, K\}} \energy(\dist(x(i), x(c(k))).
\end{equation*}

\section{The \pam{} algorithm}
\label{PAMpseudo}

\begin{algorithm}
\begin{algorithmic}[1]
\State $t \gets 0$
\State Initialize $\mathcal{C}_0 \subset \{1, \ldots, N\}$.
\While{\textbf{true}}
\For{$i_p \in \{1, \ldots, N\} \setminus \mathcal{C}_t$}
\For{$k_p \in \{1, \ldots, K\} $}
\State  $\psi_{t+1}^p(i_p, k_p) \gets \displaystyle \sum_{i=1}^N  \min_{i' \in \mathcal{C}_t \setminus \{c_t(k_p)\} \cup \{i_p\}} \energy(\dist(x(i), x(i')))$
\EndFor
\EndFor
\State $\displaystyle i_p^*, k_p^* \gets \argmin_{i_p, k_p} \psi_{t+1}^p(i_p, k_p)$ 
\If{$\psi_{t+1}^p(i_p^*, k_p^*) < 0$}
\State $\mathcal{C}_{t+1} \gets \mathcal{C}_t \setminus \{c_t(k_p^*)\} \cup \{i_p^*\}$
\Else{}
\State \textbf{break}
\EndIf
\State $t \gets t+1$
\EndWhile
\end{algorithmic}
\caption{The \pam{} algorithm of \cite{generic_kmedoids} is a computationally inefficient predecessor of \clarans{}. At lines 4 and 5, one loops over all possible (medoid, non-medoid) swaps, recording the energy obtained with each swap. At line 9, the best of all possible swaps is chosen. At line 10, if the best found swap results in a decrease in energy, proceed, otherwise stop. }
\label{alg::PAM}
\end{algorithm}

\section{\clarans{} In Detail, And How To Accelerate It}
\label{app::clarans}

We start by presenting modified notation, required to describe our optimisations of \clarans{}~\citep{clarans} in full pseudocode. As before, we will let the $N$ samples which we want to partition into $K$ clusters be $x(1), \ldots, x(N)$. Let $t \in \{1, \ldots, \infty \}$ denote the current round of the algorithm. Let $c_t(k) \in \{1, \ldots, N\}$  be the index of the sample chosen as the center of cluster $k \in \{1, \ldots, K\}$ at iteration $t$, so that $x(c_t(k))$ is the center of cluster $k$ at iteration $t$. Let $\mathcal{C}_t = \{c_t(k) \mid k \in \{1, \ldots, K\}\} \subset \{1, \ldots, N\}$ denote all such center indices. We let $a_t^1(i)$ be the cluster of sample $i$, that is
\begin{equation}
\label{defn_a}
a_t^{1}(i) = \argmin_{k \in \{1, \ldots, K\}} f(x(i),x(c_t(k))).
\end{equation}
Let $\Energy_t(k)$ denote the sum of the dissimilarities of elements in cluster $k$ at iteration $t$, also referred to as the \textit{energy} of cluster $k$, so that
\begin{align*}
\psi_t(k) &= \sum_{i : a_t^{1}(i) = k} f(x(i),x(c_t(k))). 
\end{align*}
Let $\psi_t = \sum_k \Energy_t(k)$ be the total energy, the quantity which we ultimately wish to minimize.  

We assume here that dissimilarity can be decomposed as in Eqn.~\eqref{fdecomp}, which will enable the use of the triangle inequality. 

Let $d_t^{1}(i)$ be the distance at iteration $t$ of sample $i$ to its nearest center, that is
\begin{align*}
d_t^{1}(i) &= \min_{i' \in \mathcal{C}_t} \dist(x(i),x(i')).
\end{align*}
Under assumption~\eqref{fdecomp}, we now have~\eqref{defn_a} taking the form, 
\begin{equation*}
a_t^{1}(i) = \argmin_{k \in \{1, \ldots, K\}} \dist(x(i),x(c_t(k))), 
\end{equation*}
so that $d_t^{1}(i) = \dist(x(i), x(c_t(a_t^{1}(i))))$. In the same way as we use $a_t^{1}(i)$ and $d_t^{1}(i)$ for the nearest center, we will use $a_t^{2}(i)$ and $d_t^{2}(i)$ for the second nearest center, that is
\begin{align*}
d_t^{2}(i) &= \min_{i' \in \mathcal{C}_t \setminus \{ c_t(a_t^{1}(i))\}} \dist(x(i),x(i')),  \\
a_t^{2}(i) &= \argmin_{k \in \{1, \ldots, K\} \setminus \{ a_t^{1}(i) \}} \dist(x(i),x(c_t(k))),
\end{align*}
so  that $d_t^{2}(i) = \dist(x(i), x(c_t(a_t^{2}(i))))$. The energy of a sample is now defined as the energy of the distance to its nearest center, so that at iteration $t$ the energy of sample $x(i)$ is $\energy(d_t^{1}(i))$. Finally, let the \textit{margin} of sample $i$ be defined as $m_t(i) = \energy(d^2_t(i)) - \energy(d^1_t(i))$.  Some cluster specific quantities which are required in the accelerated algorithm are,
\begin{align}
\begin{split}
N_t(k) &= | \{i :a_t^1(i) = k  \} |, \label{cqs} \\
D_t^1(k) &= \max_{i:a_t^1(i) = k} d_t^1(i),  \\
D_t^2(k) &= \max_{i:a_t^1(i) = k} d_t^2(i), \\
M^*_t(k) &= \frac{1}{N_t(k)}\sum_{i : a_t^1(i) = k} m_t(i).
\end{split}
\end{align}

The key triangle inequality results used to accelerate \clarans{} evaluations are now presented, with proofs in~\S\ref{app::accelerating}. Firstly, 
\begin{align*}
& \dist(x(i_p), x(c_t(k_p))) \ge D_t^1(k_p) + D_t^2(k_p) \implies \\
& \hspace{2em}  \mbox{ change in energy of cluster $k_p$ is  $N_t(k_p)M^*_t(k_p)$},
\end{align*} 
which says that if the new center $x(i_p)$ of cluster $k_p$ is sufficiently far from the old center $x(c_t(k_p))$, then all old elements of cluster $k_p$ will migrate to their old second nearest clusters, and so their change in energies will simply be their margins, which have already been computed. The second inequality used is,
\begin{align*}
& k  \not= k_p \;\land \dist(x(c_t(k)),  x(i_p))   \ge 2D_t^1(k) \implies \\
& \hspace{8em} \mbox{ no change in energy of cluster $k$},
\end{align*}
which states that if cluster $k$ is sufficiently far from the new center of $k_p$, there is no change in its energy as the indices of samples assigned to it do not change. 

These implications allow changes in energies of entire clusters to be determined in a single comparison. Clusters likely to benefit from these tests are those lying far from the new proposed center $x(i_p)$. The above tests involve the use of $\dist(x(c_t(k)),  x(i_p))$, but the computation of this quantity can sometimes be avoided by using the inequality, 
\begin{equation*}
\dist(x(c_t(k)), x(i_p)) \ge cc_t(a_t^1(i_p), k)  - D_t^1(i_p),
\end{equation*}
where $cc_t$ is the $K \times K$ matrix of inter-medoid distances at iteration $t$. To accelerate the update step of \clarans{}, the following bound test is used, 
\begin{align}
&  \min(\dist(x(c_t(k_p)), x(c_t(k))), \dist(x(i_p), x(c_t(k)))) \notag \\ 
& \hspace{3em} > D_t^1(k) + D_t^2(k) \implies \mbox{ no change in cluster $k$.} \notag
\end{align}
We also use a per-sample version of the above inequality for the case of failure to eliminate the entire cluster. Full proofs, descriptions, and algorithms incorporating these triangle inequalities can be found in~\ref{app::accelerating}.

\subsection{Review of notation and ideas}
Consider a proposed update for centers at iteration $t+1$, where the center of cluster $k_p$ is replaced by $x(i_p)$. Let $\delta_t(i \mid k_p \lhd i_p)$ denote the change in energy of sample $i$ under such an update, that is
\begin{align*}
\delta_t(i \mid k_p \lhd i_p) &= \mbox{ energy after swap } - \mbox{energy before swap } \\
&= \min_{i' \in \mathcal{C}_t  \setminus  \{c_t(k_p)\} \cup \{i_p\}  } \energy(\dist(x(i),x(i'))) - \energy(d_t^{1}(i)). 
\end{align*} 
We choose subscript `$p$'  for $k_p$ and $i_p$, as together they define a \emph{proposed} swap.  We will write $a^{12}d_t^{12}(i) = \{a_t^1(i), a_t^2(i), d_t^1(i), d_t^2(i)\}$ throughout for brevity. Finally, let
\begin{align*}
D_t^1(k) &= \max_{i:a_t^1(i) = k} d_t^1(i),  \\
D_t^2(k) &= \max_{i:a_t^1(i) = k} d_t^2(i).
\end{align*}

\begin{algorithm}[ht!]
\begin{algorithmic}[1]
\State Make proposal $k_p \in \{1, \ldots K\}$ and  $i_p \in \{1, \ldots, N\} \setminus \mathcal{C}_t$.
\State $\Delta_t(k_p \lhd i_p) \gets \frac{1}{N}\sum_{i = 1}^{N} \delta_t(i \mid k_p \lhd i_p )$ 
\Comment{The assignment \emph{evaluation} step, see Alg.~\ref{alg::claransdelta}}
\If{$\Delta_t < 0$}
  \State $\mathcal{C}_{t+1} \gets \mathcal{C}_{t} \setminus  \{c_t(k_p)\} \cup \{i_p\}$
  \For{$i\in \{1, \ldots, N\}$}
    \State Set $a^{12}d_{t+1}^{12}(i)$
    \Comment{The \emph{update} step, see Alg.~\ref{alg::claransupdate}}
  \EndFor
\Else
  \State $\mathcal{C}_{t+1} \gets \mathcal{C}_{t}$
  \For{$i\in \{1, \ldots, N\}$}
    \State $a^{12}d_{t+1}^{12}(i) \gets a^{12}d_t^{12}(i)$
  \EndFor
\EndIf
\end{algorithmic}
\caption{One round of \clarans{}. The potential bottlenecks are the proposal \emph{evaluation} at line 2 and the \emph{update} at line 6. The cost of proposal evaluation, if all distances are pre-computed, is $O(N)$, while if distances are not pre-computed it is $O(dN)$ where $d$ is the cost of a distance computation. As for the update step, there is no cost if $\Delta_t \ge 0$ as nothing changes, however if the proposal is accepted then $\mathcal{C}_{t+1} \not= \mathcal{C}_t$, and all data whose nearest or second nearest center change needs updating.}
\label{alg::claransstep}
\end{algorithm}

\begin{algorithm}[ht!]
\begin{algorithmic}[1]
\State $d_{} \gets \dist(x(i), x(i_p))$
\If{$a_t^1(i) = k_p$}
\If{$d_{} \ge d_t^2(i)$}
\State $\delta_t(i \mid k_p \lhd i_p ) \gets \energy(d_t^2(i)) - \energy(d_t^1(i))$
\Else 
\State $\delta_t(i \mid k_p \lhd i_p ) \gets \energy(d_{}) - \energy(d_t^1(i))$
\EndIf
\Else
\If{$d_{} \ge d_t^1(i)$}
\State $\delta_t(i \mid k_p \lhd i_p ) \gets 0$
\Else
\State $\delta_t(i \mid k_p \lhd i_p ) \gets \energy(d_{}) - \energy(d_t^1(i))$
\EndIf
\EndIf
\end{algorithmic}
\caption{Standard approach (level 0) with \clarans{} for computing $\delta_t(i \mid k_p \lhd i_p )$ at iteration $t$, as described in~\cite{clarans}. Note however that here we do not store all $N^2$ distances, as in~\cite{clarans}. }
\label{alg::claransdelta}
\end{algorithm}

\begin{algorithm}[ht!]
\begin{algorithmic}[1]
\Acomm{If the center which moves is nearest or second nearest, complete update required} 
\If{$a_t^1(i) = k_p$ or $a_t^2(i) = k_p$}
\State Get $\dist(x(i), x(c_{t+1}(k)))$ for all $k \in \{1, \ldots, K\}$
\State Use above $k$ distances to set $a^{12}d_{t+1}^{12}(i)$
\Else
\Acomm{$d_{t}^1(i)$ and $d_{t}^2(i)$ are still valid distances, so need only check new candidate center $k_p$}
\State $d_{} \gets \dist(x(i), x(i_p))$
\State Use the fact that $\{d_{t+1}^1(i), d_{t+1}^2(i)\} \subset \{d_{t}^1(i), d_{t}^2(i), d\}$ to set $a^{12}d_t^{12}(i)$
\EndIf
\end{algorithmic}
\caption{Simple approach (level 0) with \clarans{} for computing  $a^{12}d_{t+1}^{12}(i)$}
\label{alg::claransupdate}
\end{algorithm}

\subsection{Accelerating \clarans{}}
\label{app::accelerating}
We now discuss in detail how to accelerate the proposal evaluation and the cluster update. We split our proposed accelerations into 3 levels. At levels 1 and 2, triangle inequality bounding techniques are used to eliminate distance calculations. At level 3, an early breaking scheme is used to quickly reject unpromising swaps.

\subsubsection{Basic triangle inequalities bounds}
\label{sec:intro}
We show how $\delta_t(i \mid k_p \lhd i_p)$ can be bounded, with the final bounding illustrated in Figure~\ref{hoeffding1}. There are four bounds to consider : upper and lower bounds for each of the two cases $k_p = a_t^1(i)$ (the center being replaced is the center of element $i$) and $k_p \not= a_t^1(i)$ (the center being replaced is not the center of element $i$).  We will derive a lower bound for the two cases simultaneously, thus we will derive 3 bounds. First, consider the upper bound for the case $k_p \not= a_t^1(i)$,
\begin{align*}
\delta_t(i \mid k_p \lhd i_p) &= \min_{i' \in \mathcal{C}_t  \setminus  \{c_t(k_p)\} \cup \{i_p\}  } \energy(\dist(x(i),x(i'))) - \energy(d_t^{1}(i)), \\
&= \min_{i' \in \{c_t(a_t^1(i)), i_p\}} \energy(\dist(x(i),x(i'))) - \energy(d_t^{1}(i)), \\
&\le \energy(\dist(x(i),x(c_t(a_t^1(i))))) - \energy(d_t^{1}(i)), \\
& = 0.
\end{align*}
and thus we have
\begin{equation}
\label{eqn:drop}
k_p \not= a_t^1(i) \implies \delta_t(i \mid k_p \lhd i_p) \le 0.
\end{equation}
Implication~\ref{eqn:drop} simply states the obvious fact that the energy of element $i$ cannot increase when a center other than that of cluster $a_t^1(i)$ is replaced. The other upper bound case to consider is $k_p = a_t^1(i)$, which is similar,
\begin{align*}
\delta_t(i \mid k_p \lhd i_p) &= \min_{i' \in \mathcal{C}_t  \setminus  \{c_t(k_p)\} \cup \{i_p\}  } \energy(\dist(x(i),x(i'))) - \energy(d_t^{1}(i)), \\
&= \min_{i' \in \{c_t(a_t^2(i)), i_p\}} \energy(\dist(x(i),x(i'))) - \energy(d_t^{1}(i)), \\
&\le \energy(\dist(x(i),x(c_t(a_t^2(i))))) - \energy(d_t^{1}(i)), \\
&= \energy(d_t^{2}(i))  - \energy(d_t^{1}(i)), \\
& = m_t(i), \\
& \le M_t(k_p).
\end{align*}
and thus we have 
\begin{equation}
\label{eqn:fallbackbound}
k_p = a_t^1(i) \implies \delta_t(i \mid k_p \lhd i_p) \le M_t(i).
\end{equation}
Implication~\ref{eqn:fallbackbound} simply states the energy of element $i$ cannot increase by more than the maximum margin in the cluster of $i$ when it is the center of cluster $a_t^1(i)$ which is replaced. We now consider lower bounding $\delta_t(i \mid k_p \lhd i_p)$ for both the cases $a_t^1(i) = k_p$ and $a_t^1(i) \not= k_p$ simultaneously. We choose to bound them simultaneously as doing so separately arrives at the same bound.
\begin{align}
\delta_t(i \mid k_p \lhd i_p) &= \min_{i' \in \mathcal{C}_t  \setminus  \{c_t(k_p)\} \cup \{i_p\}  } \energy(\dist(x(i),x(i'))) - \energy(d_t^{1}(i)), \notag \\
 &\ge \min_{i' \in \mathcal{C}_t  \cup \{i_p\}  } \energy(\dist(x(i),x(i'))) - \energy(d_t^{1}(i)), \notag \\
&= \min_{i' \in \{c_t(a_t^1(i)), i_p\}} \energy(\dist(x(i),x(i'))) - \energy(d_t^1(i)), \notag \\
&=\min\left(0, \energy(\dist(x(i),x(i_p))) - \energy(d_t^1(i)) \right), \notag \\
&\ge\min\left(0, \energy(\dist(x(i),x(i_p))) - \energy(D_t^1(a_t^1(i))) \right).\label{bla1}
\end{align}
Let $d_p(k)$ denote the distance between the elements in the proposed swap,
\begin{equation*}
d_p(k)= \dist(x(c_t(k)),x(i_p)).
\end{equation*}
The triangle inequality guarantees that, 
\begin{equation}
\label{tri1}
\dist(x(i),x(i_p)) \ge \begin{cases}
0 &\text{if $d_p(a_t^1(i)) \le D_t^1(a_t^1(i)) $},\\
d_p(a_t^1(i)) - D_t^1(a_t^1(i)) &\text{if $D_t^1(a_t^1(i)) < d_p(a_t^1(i))  $}.
\end{cases}
\end{equation}
Using~\eqref{tri1} in~\eqref{bla1} we obtain,
\begin{equation}
\label{blaze}
\delta_t(i \mid k_p \lhd i_p) \ge \begin{cases}
-\energy(D_t^1(a_t^1(i))) &\text{if $d_p(a_t^1(i)) \le D_t^1(a_t^1(i)) $} \\
\energy(d_p(a_t^1(i)) - D_t^1(a_t^1(i))) - \energy(D_t^1(a_t^1(i))) &\text{if  $D_t^1(a_t^1(i)) < d_p(a_t^1(i)) \le 2D_t^1(a_t^1(i)) $} \\
0 &\text{if $2D_t^1(a_t^1(i)) < d_p(a_t^1(i)).$} 
\end{cases} 
\end{equation}
These are the lower bounds illustrated in Figure~\ref{hoeffding1}. Define $\Delta(k \mid k_p \lhd i_p)$ to be the average change in energy for cluster $k$ resulting from a proposed swap, that is,
\begin{equation*}
\Delta_t(k \mid k_p \lhd i_p) = \frac{1}{N_t(k)}\sum_{i : a_t^1(i) = k_p} \delta_t(i \mid k_p \lhd i_p).
\end{equation*}
Let the average of the change in energy over all data resulting from a proposed swap be $\Delta_t(k_p \lhd i_p)$, that is
\begin{equation*}
\Delta_t(k_p \lhd i_p) = \sum_k p_t(k) \Delta_t(k \mid k_p \lhd i_p).
\end{equation*}
One can show that for $k = k_p$,
\begin{equation}
\label{tail:kp}
d_p(k_p) \ge D_t^1(k_p) + D_t^2(k_p) \implies \Delta_t(k \mid k_p \lhd i_p) =  \,M^*_t(k_p).
\end{equation}
The equality~\eqref{tail:kp} corresponds to a case where the proposed center $x(i_p)$ is further from every point in cluster $k_p$ than is the second nearest center, in which case the increase in energy of cluster $k_p$ is simply the sum of margins. It corresponds to the solid red horizontal line in Figure~\ref{hoeffding1}, left.

\begin{figure}
\begin{center}
\includegraphics[width=\textwidth]{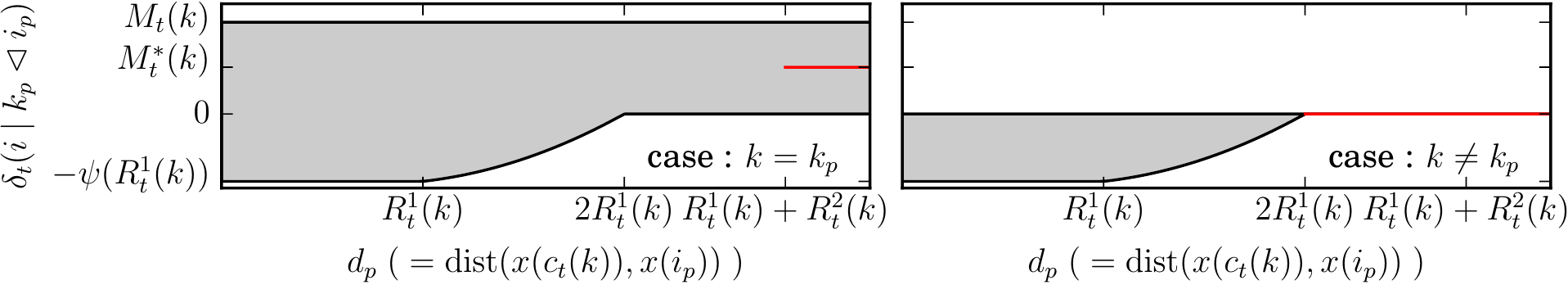}
\caption{Illustrating the bounds. Dark gray regions denote possible changes in energy of elements. On the left, the case $k = k_p$, where the solid line segment is the average change in element energy in the case where $d_p$ exceeds a certain radius. On the right, the case $k \not = k_p$, where sample energies can only decrease.  }
\label{hoeffding1}
\end{center}
\vskip -0.2in
\end{figure}

\subsubsection{Level 1 proposal evaluation accelerations}
\label{sec:level1}
What we wish to evaluate when considering a proposal is the mean change in energy, that is, 
\begin{equation}
\label{decomp1}
\frac{1}{N}\sum_{i = 1}^{N} \delta_t(i \mid k_p \lhd i_p ) =\frac{1}{N}\Bigg(  \underbrace{\sum_{k:k\not=k_p} \sum_{i : a_t^{1}(i) = k} \delta_t(i \mid k_p \lhd i_p )}_{\left(N - N_t(k_p)\right)\Delta_t^{-}(k_p \lhd \,i_p)} + \underbrace{\sum_{i : a_t^{1}(i) = k_p} \delta_t(i \mid k_p \lhd i_p )}_{N_t(k_p)\Delta(k_p \mid k_p \lhd \,i_p)} \Bigg).
\end{equation}
Where in~\eqref{decomp1} we define $\Delta_t^{-}(k_p \lhd \,i_p)$ as,
\begin{equation*}
\Delta_t^{-}(k_p \lhd \,i_p) = \frac{1}{N - N_t(k_p)}\sum_{k:k\not=k_p} \sum_{i : a_t^{1}(i) = k} \delta_t(i \mid k_p \lhd i_p ).
\end{equation*}
From~\ref{sec:intro} we have the result, corresponding to the solid line in Figure~\ref{hoeffding1}, that
\begin{equation}
\label{zoombaf}
a_t^1(i) = k \land k \not= k_p \land \dist(x(c_t(k)), x(i_p))  \ge 2D_t^1(k) \implies \delta_t(i \mid k_p \lhd i_p ) = 0.
\end{equation}
We use this result to eliminate entire clusters in the proposal evaluation step: a cluster $k$ whose center lies sufficiently far from $x(i_p)$ will not contribute, as long as $k \not=k_p$,  
\begin{equation*}
\Delta_t^{-}(k_p \lhd \,i_p) = \frac{1}{N - N_t(k_p)}\sum_{\substack{k:k\not=k_p \land \\ \dist(x(i_p), x(c_t(k)))  < 2D_t^1(k)}} \sum_{i : a_t^{1}(i) = k} \delta_t(i \mid k_p \lhd i_p ).
\end{equation*}
Implication~\ref{tail:kp}, corresponding to the solid line in Figure~\ref{hoeffding1}, left, can be used in the case $k = k_p$ to rapidly obtain the second term in~\eqref{decomp1} if $\dist(x(c_t(k_p)), x(i_p)) \ge D_t^1(k) + D_t^2(k)$. 

The level 1 techniques for obtaining whole cluster sums require the distances from $x(i_p)$ to all cluster centers, although in \S \ref{level2} (level 2) we show how even these distance calculations can sometimes be avoided. A second layer of element-wise triangle inequality tests is included for the case where the test on an entire cluster fails. 

These level 1 techniques for accelerating the proposal are presented in Alg.~\ref{alg:TRICLARANS1:eval}.

\begin{algorithm}
\begin{algorithmic}[1]
\Acomm{Set distances from proposed center $x(i_p)$ to all current centers $\mathcal{C}_t$}
\Comment{$O(dK)$}
\For{$k \in \{1, \ldots, K\}$}
\State $d_c(k) \gets \dist(x(i_p), x(c_t(k)))$
\EndFor 
\State $d_{pp} \gets d_c(k_p)$
\Acomm{Process cluster $k_p$}
\Comment{$O(dN/K^2)$}
\State \texttt{CLARANS-12-EVAL-P}$()$
\Acomm{Process all other clusters } 
\Comment{$O(dN/K)$}
\State \texttt{CLARANS-1-EVAL-N-P}$()$
\end{algorithmic}
\caption{\texttt{CLARANS-1-EVAL} : proposal evaluation using level 1 accelerations. We call subroutines for processing the cluster $k_p$ (\texttt{CLARANS-12-EVAL-P}) and all other clusters (\texttt{CLARANS-1-EVAL-N-P}). The expected complexity for the full evaluation is $O(d(K + N/K))$. The expected complexity for \texttt{CLARANS-12-EVAL-P} assumes that the probability that cluster $k_p$ is not processed using~\eqref{tail:kp} is $O(1/K)$.}
\label{alg:TRICLARANS1:eval}
\end{algorithm}

\begin{algorithm}
\begin{algorithmic}[1]
\Acomm{Try to use \eqref{tail:kp} to quickly process cluster $k_p$}
\If{$d_{pp}  \ge D_t^1(k_p) +  D_t^2(k_p)$}
\State $\Delta_t(k_p \lhd i_p) = \Delta_t(k_p \lhd i_p) + p_t(k_p) M_t^*(k_p)$ 
\Else
\Acomm{Test~\eqref{tail:kp} failed, enter element-wise loop for cluster $k_p$}
\For{$i \in \{i' : a_t^1(t') = k_p$\}}
\Acomm{Try tighter element-wise version of \eqref{tail:kp} to prevent computing a distance}
\If{$d_{pp} \ge d_t^1(i) +  d_t^2(i)$}
\State $\Delta_t(k_p \lhd i_p) = \Delta_t(k_p \lhd i_p) + m_t(i)/N$ 
\Else{}
\Acomm{Test failed, need to compute distance}
\State $d \gets \dist(x(i_p), x(i))$
\State $\Delta_t(k_p \lhd i_p) = \Delta_t(k_p \lhd i_p) + \min(d, m_t(i))/N$ 
\EndIf
\EndFor
\EndIf
\end{algorithmic}
\caption{\texttt{CLARANS-12-EVAL-P} : adding the contribution of cluster $k_p$ to $\Delta_t(k_p \lhd i_p)$. The key inequality here is~\eqref{tail:kp}, which states that if $i_p$ is sufficiently far from the center of cluster $k_p$, then elements in cluster $k_p$ will go to their current second nearest center if the center of $k_p$ is removed.}
\label{alg:TRICLARANS1:eval_pr}
\end{algorithm}

\begin{algorithm}
\begin{algorithmic}[1]
\For{$k \in \{1, \ldots, K\}  \setminus \{k_p\}$}
\Acomm{Try to use \eqref{zoombaf} to quickly process cluster $k$}
\If{$d_c(k) < 2D_t^1(k)$}
\Acomm{Test~\eqref{zoombaf} failed, enter element-wise loop for cluster $k$}
\For{$i \in \{i' : a_t^1(t') = k$\}}
\Acomm{Try tighter element-wise version of \eqref{zoombaf} to prevent computing a distance}
\If{$d_c(k) < 2d_t^1(i)$}
\Acomm{Test failed, need to compute distance}
\State $d \gets \dist(x(i_p), x(i))$
\If{$d < d_t^1(i)$}
\State $\Delta_t(k_p \lhd i_p) = \Delta_t(k_p \lhd i_p) + (d - d_t^1(i))/N$
\EndIf
\EndIf
\EndFor
\EndIf
\EndFor
\end{algorithmic}
\caption{\texttt{CLARANS-1-EVAL-N-P} : adding contributions of all clusters $k \not= k_p$ to  $\Delta_t(k_p \lhd i_p)$. The key inequality used is~\eqref{zoombaf}, which states that if the distance between $x(i_p)$ and the center of cluster $k$ is large relative to the distance from the center of cluster $k$ to its most distant member, then there is no change in energy in cluster $k$. }
\label{alg:TRICLARANS1:eval_npr}
\end{algorithm}

\subsubsection{Level 1 cluster update accelerations}
If a proposal is accepted, the standard \texttt{CLARANS} uses Alg.~\eqref{alg::claransupdate} to obtain  $a^{12}d_{t+1}^{12}(i)$, where every element $i$ requires at least $1$ distance calculation, with those elements for which cluster $k_p$ is the nearest or second nearest at $t$ require $K$ distance calculations. Here at level 1, we show how many samples requiring $1$ distance calculation can be set without any distance calculations, and even better: how entire clusters can sometimes be processed in constant time.

The inequality to eliminate an entire cluster is,
\begin{align}
\begin{split}
\min(\dist(x(c_t(k_p)), x(c_t(k))), \dist(x(i_p), x(c_t(k)))) &> D_t^1(k) + D_t^2(k) \label{grossoumodo} \\
&\implies \mbox{ no change in cluster $k$.} 
\end{split}
\end{align}

While the inequality used to eliminate the distance calculation for a single sample is,
\begin{align}
\begin{split}
\min(\dist(x(c_t(k_p)), x(c_t(k))), \dist(x(i_p), x(c_t(k)))) &> d_t^1(i) + d_t^2(i) \label{finerat} \\
&\implies \mbox{ no change for sample $i$.}
\end{split}
\end{align}
Note that the inequalities need to be strict, `$\ge$' would not work. The test~\eqref{grossoumodo} is illustrated in Figure~\ref{fig:nochange}, left. These bound tests are used in Alg.~\eqref{alg:TRICLARANS1:update}. The time required to update cluster related quantities ($D_1, D_t, M^*$) is negligible as compared to updating sample assignments, and we do not do anything clever to accelerate it, other than to note that only clusters which fail to be eliminated by~\eqref{grossoumodo} potentially require updating.

\begin{figure}
\begin{center}
\includegraphics[width=0.4\textwidth]{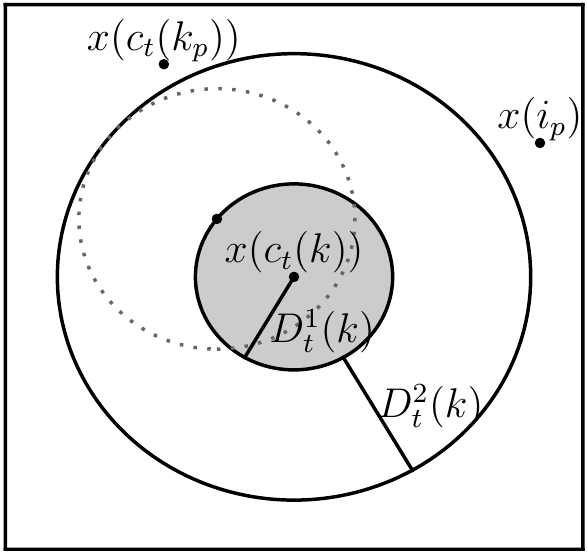}
\caption{(left) illustrating test~\ref{grossoumodo}. Consider an element $x(i)$ with $a_t^1(i) = k$, so that $x(i)$ lies in the inner gray circle, and $k \not= k_p$. Firstly, $\dist(x(c_t(k_p)), x(c_t(k))) > D_t^1(k) + D_t^2(k)$ means that $a_t^2(i) \not= k_p$, thus both $d_t^1(i)$ and $d_t^2(i)$ will be valid distances at iteration $t+1$. Then, as $\dist(x(i_p), x(c_t(k))) > D_t^1(k) + D_t^2(k)$, we have $\dist(x(i_p), x(i)) > D_t^2(k) \ge d_t^2(i)$, so $a_{t+1}^2(i) = a_{t}^2(i)$. }
\label{fig:nochange}
\end{center}
\vskip -0.2in
\end{figure}

\begin{algorithm}
\begin{algorithmic}[1]
\Acomm{Set distance from centers to the nearer of new and old cluster center $k_p$}
\Comment{$O(dK)$}
\For{$k \in \{1, \ldots, K\}$}
\State $d_c(k) \gets \min(\dist(x(c_t(k_p)), x(c_t(k))), \dist(x(i_p), x(c_t(k))))$
\State $\;\;\;\;\;\;\;\;\left(\, = \min(\dist(x(c_t(k_p)), x(c_t(k))), \dist(x(c_{t+1}(k_p)), x(c_t(k)))) \right)$
\EndFor 
\Acomm{Process elements in cluster $k_p$ from scratch}
\For{$i \in \{i' : a_t^1(t') = k_p$\}}
\State Obtain $a^{12}d_{t+1}^{12}(i)$ from scratch
\EndFor{}
\Acomm{Process all other clusters}
\For{$k \in \{1, \ldots, K\}  \setminus \{k_p\}$}
\Acomm{Try to use \eqref{grossoumodo} to quickly process cluster $k$}
\If{$d_c(k) \le D_t^1(k) + D_t^2(k)$}
\For{$i \in \{i' : a_t^1(t') = k$\}}
\Acomm{Try to use \eqref{finerat} to quickly process element $i$}
\If{$d_c(k) \le d_t^1(k) + d_t^2(k)$}
\If{$a_t^2(i) = k_p$}
\State Obtain $a^{12}d_{t+1}^{12}(i)$ from scratch
\Else{}
\State $d \gets \dist(x(i),x(i_p))$
\State Use $\{d_{t+1}^1(i), d_{t+1}^2(i)\} \subset \{d_{t}^1(i), d_{t}^2(i), d\}$ as in~\eqref{alg::claransupdate}
\EndIf
\Else{}
\State $a^{12}d_{t+1}^{12}(i) \gets a^{12}d_t^{12}(i)$
\EndIf
\EndFor
\EndIf
\EndFor
\State Update cluster statistics for $t+1$ where necessary
\end{algorithmic}
\caption{\texttt{CLARANS-1-UPDATE} : cluster update using level 1 accelerations. Inequalities~\eqref{grossoumodo} and~\eqref{finerat} are used to accelerate the updating of $a^{12}d_{t+1}^{12}(i)$ for $i : a^1(i) \not= k_p \land a^2(i) \not= k_p$. Essentially these inequalities say that if neither the old center of cluster $k_p$ nor its new center $x(i_p)$ are near to an element (or all elements in a cluster), then the nearest and second element of that element (or all elements on a cluster) will not change.}
\label{alg:TRICLARANS1:update}
\end{algorithm}

\subsubsection{Level 2 proposal evaluation accelerations}
\label{level2}

We now discuss level 2 accelerations. Note that these accelerations come at the cost of an increase of $O(K^2)$ to the memory footprint. The key idea is to maintain all $K^2$ inter-center distances, denoting by $cc_t(k, k') = \dist(x(c_t(k)), x(c_t(k')))$ the distance between centers of clusters $k$ and $k'$. At level 1, all distances $\dist(x(i_p), x(c_t(k)))$ for $k \in \{1, \ldots, K\}$ are computed up-front for proposal evaluation, but here at level 2 we use,
\begin{equation}
\label{quickjob}
\dist(x(i_p), x(c_t(k))) \ge cc_t(a_t^1(i_p), k)  - d_t^1(i_p),
\end{equation}
to eliminate the need for certain of these distances. Combining~\eqref{quickjob} with~\eqref{zoombaf} gives, 
\begin{equation}
\label{lazyzoombaf}
a_t^1(i) = k \land k \not= k_p \land cc_t(a_t^1(i_p), k)  - d_t^1(i_p)  \ge 2D_t^1(k) \implies \delta_t(i \mid k_p \lhd i_p ) = 0.
\end{equation}

\begin{algorithm}
\begin{algorithmic}[1]
\For{$k \in \{1, \ldots, K\}  \setminus \{k_p\}$} 
\Acomm{Try to use \eqref{lazyzoombaf} to quickly process cluster $k$}
\If{$cc_t(a_t^1(i_p), k)  - 2D_t^1(k)  < d_t^1(i_p) $}
\Acomm{Test~\eqref{lazyzoombaf} failed, computing $d_c(k)$ and resorting to level 1 accelerations...}
\State $d_{pk} \gets \dist(x(i_p), x(c_t(k)))$
\If{$d_{pk} < 2D_t^1(k)$}
\Acomm{Test~\eqref{tail:kp} failed, enter element-wise loop for cluster $k$}
\For{$i \in \{i' : a_t^1(t') = k$\}}
\Acomm{Try tighter element-wise version of \eqref{zoombaf} to prevent computing a distance}
\If{$d_{pk} < 2d_t^1(i)$}
\Acomm{Test failed, need to compute distance}
\State $d \gets \dist(x(i_p), x(i))$
\If{$d < d_t^1(i)$}
\State $\Delta_t(k_p \lhd i_p) \gets \Delta_t(k_p \lhd i_p) + (d - d_t^1(i))/N$
\EndIf
\EndIf
\EndFor
\EndIf
\EndIf
\EndFor
\end{algorithmic}
\caption{\texttt{CLARANS-2-EVAL-N-P} : add contribution of all clusters $k \not= k_p$ to  $\Delta_t(k_p \lhd i_p)$. In addition to the bound tests used at level 1, inequality \eqref{lazyzoombaf} is used to test if a center-center distance needs to be calculated.}
\label{alg:TRICLARANS2:eval_npr}
\end{algorithm}

\begin{algorithm}
\begin{algorithmic}[1]
\State $d_{pp} \gets \dist(x(i_p), x(c_t(k_p)))$
\Acomm{Process cluster $k_p$}
\Comment{$O(dN/K^2)$}
\State \texttt{CLARANS-12-EVAL-P}$()$
\Acomm{Process all other clusters } 
\Comment{$O(dN/K)$}
\State \texttt{CLARANS-2-EVAL-N-P}$()$
\end{algorithmic}
\caption{\texttt{CLARANS-2-EVAL} : Proposal evaluation using level 2 accelerations. Unlike at level 1, not all distances from $x(i_p)$ to centers need to be computed up front.} 
\label{alg:TRICLARANS2:eval}
\end{algorithm}

\subsubsection{Level 2 cluster update accelerations}
The only acceleration added at level 2 for the cluster update is for the case $k_p \in \{a_t^1(i), a_t^2(i) \}$, where at level 1, $a^{12}d_{t+1}^{12}$ is set from scratch, requiring all $K$ distances to centers to be computed. At level 2, we use $cc_t$ to eliminate certain of these distances using \texttt{WARMSTART}, which takes in the distances to 2 of the $K$ centers and uses the larger of these as a threshold beyond which any distance to a center can be ignored.


\begin{algorithm}
\begin{algorithmic}[1]
\State For $k \in \{1, \ldots, K\} \setminus \{k_p\}$: compute $\dist(x(i_p), x(c_t(k))) \,\left( \, = \dist(x(c_{t+1}(k_p)), x(c_{t+1}(k))  \right)$  and set $cc_{t+1}$ accordingly (in practice we don't need to store $cc_{t}$ and $cc_{t+1}$ simultaneously as they are very similar).
\For{$k \in \{1, \ldots, K\}$} 
\State $d_c(k) \gets \min(cc_{t}(k_p, k), cc_{t+1}(k_p, k))$
\EndFor 
\Acomm{Process elements in cluster $k_p$ from scratch}
\For{$i \in \{i' : a_t^1(t') = k_p$\}}
\State $d \gets \dist(x(i),c_{t+1}(k_p))$
\State Obtain $a^{12}d_{t+1}^{12}(i)$, using \texttt{WARMSTART} with $d$ and $d_t^2(i)$.
\EndFor{}
\Acomm{Process all other clusters}
\For{$k \in \{1, \ldots, K\}  \setminus \{k_p\}$}
\Acomm{Try to use \eqref{grossoumodo} to quickly process cluster $k$}
\If{$d_c(k) \le D_t^1(k) + D_t^2(k)$}
\For{$i \in \{i' : a_t^1(t') = k$\}}
\Acomm{Try to use \eqref{finerat} to quickly process element $i$}
\If{$d_c(k) \le d_t^1(k) + d_t^2(k)$}
\State $d \gets \dist(x(i),c_{t+1}(k_p))$
\If{$a_t^2(i) = k_p$}
\State Obtain $a^{12}d_{t+1}^{12}(i)$, using \texttt{WARMSTART} with $d$ and $d_t^1(i)$.
\Else{}
\State Use $\{d_{t+1}^1(i), d_{t+1}^2(i)\} \subset \{d_{t}^1(i), d_{t}^2(i), d\}$ as in~\eqref{alg::claransupdate}
\EndIf
\Else{}
\State $a^{12}d_{t+1}^{12}(i) \gets a^{12}d_t^{12}(i)$
\EndIf
\EndFor
\EndIf
\EndFor
\State Update cluster statistics for $t+1$ where necessary
\end{algorithmic}
\caption{\texttt{CLARANS-2-UPDATE} : update using level 2 accelerations. The only addition to level 1 accelerations is the use of \texttt{WARMSTART} to avoid computing all $k$ sample-center distances for elements whose nearest or second nearest is $k_p$.}
\label{alg:TRICLARANS2:update}
\end{algorithm}


\subsection{Level 3}
\label{sec:level3}
At levels 1 and 2, we showed how \clarans{} can be accelerated using the triangle inequality. The accelerations were exact, in the sense that for a given initialization, the clustering obtained using \clarans{} is unchanged whether or not one uses the triangle inequality. 

Here at level 3 we diverge from exact acceleration. In particular, we will occasionally reject good proposals. However, the proposals which are accepted are still only going to be good ones, so that the energy strictly decreases. In this sense, it is not like stochastic gradient descent, where the loss is allowed to increase.

The idea is to the following. Given a proposal swap : replace the center of cluster $k_p$ with the element indexed by $i_p$, use a small sample of data to estimate the quality of the swap, and if the estimate is bad (increase in energy) then immediately abandon the proposal and generate a new proposal. If the estimate is good, obtain a more accurate estimate using more ($2\times$) elements. Repeat this until all the elements have been used and the exact energy under the proposed swap is known : if the exact energy is lower, implement the swap otherwise reject it. 

The level 1 and 2 accelerations can be used in parallel with the acceleration here. The elements sub sampled at level 3 are chosen to belong to clusters which are not eliminated using level 1 and 2 cluster-wise bound tests.  Suppose that there are $\tilde{K}$ clusters which are not eliminated at level 2, we choose the number of elements chosen in the smallest sub sample to be $30\tilde{K}$. Thereafter the number of elements used to estimate the post-swap energy doubles. 

Let the number of elements in the $\tilde{K}$ non-eliminated clusters by $n_{A}$ and the number sampled be $n_{S}$, so that $n_{S} = 30\tilde{K}$. Supposing that $n_A/n_{S}$ is a power of 2. Then, one can show that the probability that a good swap is rejected is bounded above by $1 - n_{S}/n_{A}$. Consider the case $n_A/n_{S} = 2$, so that the sample is exactly half of the total. Suppose that the swap is good. Then, if the sum over the sample is positive, the sum over its complement must be negative, as the total sum is negative. Thus there at least as many ways to draw $n_S$ samples whose sum is negative as positive. 

If $n_A/n_{S} = 4$, then consider what happens if one randomly assign another quarter to the sample. With probability one half the sum is negative, thus by the same reasoning with probability at least $1/2 \times 1/2 = 1/4$ the sum over the original $n_S$ samples is negative.

\begin{algorithm}
\begin{algorithmic}[1]
\State Determine which clusters are not eliminated at level 2, define to be $U$. 
\State $\tilde{K} \gets |U|$.
\State $N_T \gets \sum_{k \in U} N_t(k)$
\State $N_S \gets 30\tilde{K}$
\State $S \gets $ uniform sample of indices of size $N_S$ from clusters $U$
\State $\hat{\Delta}_t(k_p \lhd i_p) \gets \infty$
\While{$N_S < N_T$ and $\hat{\Delta}_t(k_p \lhd i_p) < 0$}
\State $\hat{\Delta}_t(k_p \lhd i_p) \gets \frac{1}{N_S} \sum_{i \in S}  \delta_t(i \mid k_p \lhd i_p)$
\State $N_S \gets \min(N_T, 2N_S)$
\State $S \gets S \, \cup$ uniform sample of indices so that $|S| = N_S$.
\EndWhile
\If{$N_S < N_T$}
\Return \texttt{reject}
\Else{}
\State Compute $\Delta_t(k_p \lhd i_p)$
\If{$\Delta_t(k_p \lhd i_p) < 0$}
\Return \texttt{accept}
\Else{}
\Return \texttt{reject}
\EndIf
\EndIf
\end{algorithmic}
\caption{Level 3: Schemata of using sub sampling to quickly eliminate unpromising proposals without computing an exact energy. This allows for more rapid proposal evaluation.}
\label{alg:subsampling}
\end{algorithm} 

\section{Links to datasets}
\label{urls}
The rcv1 dataset : \url{http://jmlr.csail.mit.edu/papers/volume5/lewis04a/a13-vector-files/lyrl2004_vectors_train.dat.gz}

Chromosone 10 : \url{http://ftp.ensembl.org/pub/release-77/fasta/homo_sapiens/dna/Homo_sapiens.GRCh38.dna.chromosome.10.fa}

English word list : \url{https://github.com/dwyl/english-words.git}

\section{Local minima formalism}
\begin{theorem}
A local minimum of \clarans{} is always a local minimum of \vik{}. However, there exist local minima of \vik{} which are not local minima of \clarans{}. 
\end{theorem}
\begin{proof}
The second statement is proven by the existence of example in the Introduction. For the first statement, suppose that a configuration is a local minimum of \clarans{}, so that none of the $K(N-K)$ possible swaps results in a decrease in energy. Then, each center must be the medoid of its cluster, as otherwise we could swap the center with the medoid and obtain an energy reduction. Therefore the configuration is a minimum of \clarans{}.
\end{proof}

\section{Efficient Levenshtein distance calculation}
The algorithm we have developed relies heavily on the triangle inequality to eliminate distances. However, it is also possible to abort distance calculations once started if they exceed a certain threshold of interest. When we wish to determine the 2 nearest centers to a sample for example, we can abort a distance calculation as soon as we know the distance being calculated is greater than at least two other centers. 

For vectorial data, this generally does not result in significant gains. However, when computing the Levenshtein distance it can help enormously. Indeed, for a sequence of length $l$, without a threshold on the distance the computation cost of the distance is $O(l^2)$. With a threshold $m$ it becomes $(lm)$. Essentially, only the diagonal of is searched while running the dynamic Needleman-Wunsch algorithm. We use this idea at all levels of acceleration.

\section{A Comment on Similarities used in Bioinformatics}
A very popular similarity measure in bioinformatics is that of Smith-Waterman. The idea is that similarity should be computed based on the most similar regions of sequences, and not on the entire sequences. Consider for example, the sequences $a = 123123 898989, b = 454545 898989, c = 123123 012012$. According to Smith-Waterman, these should have $sim ( a  , b ) = sim ( a  , c )  \gg sim ( b  , c )$. This is not possible to turn into a proper distance, as one would need $\dist(a,b) = \dist (a,c) \ll \dist(b,c)$, which is going to break the triangle inequality. Thus, the triangle inequality accelerations introduced cannot be applied to similarities of the Smith-Waterman type.

\section{Pre-initialising with \kmeanspp{}}
\label{sec::precede2}
In Figure~\ref{fig::precede}, we compare \kmeanspp{}+\clarans{}+\lloyd{} and \kmeanspp{}+\lloyd{}. 
\begin{figure}
\begin{minipage}[c]{0.5\textwidth}
\includegraphics[width=1.0\columnwidth]{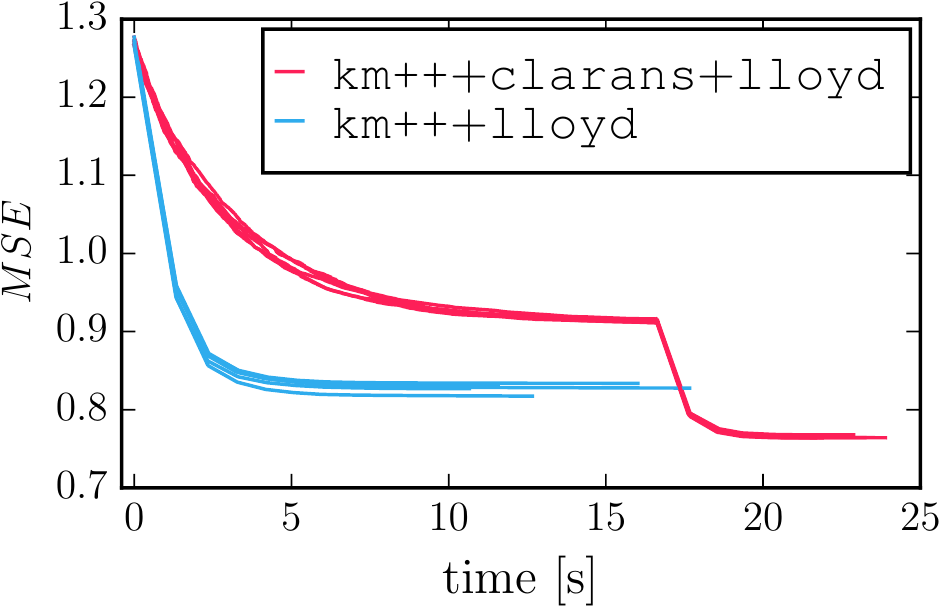}
\end{minipage}\hfill
\begin{minipage}[c]{0.48\textwidth}
\caption{Comparing \kmeanspp{}+\clarans{}+\lloyd{} and \kmeanspp{}+\lloyd{}, over ten runs, on the complete rna dataset at \protect\url{https://www.csie.ntu.edu.tw/~cjlin/libsvmtools/datasets/binary.html\#cod-rna} with dimensions $N$ = $488,565$, $d$ = $8$, and $K$ = $2,000$. We ignore here the time to run \kmeanspp{}, so that at $t$ = $0$ \kmeanspp{} has finished. Running \clarans{} before \lloyd{} results in mean final MSE of $0.76$, a significant improvement over $0.83$ obtained without \clarans{}. With \clarans{}+\lloyd{}, that is without pre-initialising with \kmeanspp{}, the mean MSE is also $0.76$, although \clarans{} runs for 28 seconds, as opposed to 18 seconds with \kmeanspp{}+\clarans{}+\lloyd{}.}
\vspace{2em}
\label{fig::precede}
\end{minipage}
\end{figure}

\section{Comapring the different optimisations levels, and kmlocal}
\label{sec:compo}
We briefly present results of the optimisations at each of the levels, as well as compare to the \clarans{} implementation accompanying \cite{Kanungo_2002_local}, an algorithm which they call `Swap'. The source code of \cite{Kanungo_2002_local} can be found at \url{https://www.cs.umd.edu/~mount/Projects/KMeans/} and is called `kmlocal', and our code is currently at \url{https://github.com/anonymous1331/km4kminit}. To the best of our knowledge, we compiled kmlocal correctly, and used the default -O3 flag in the Makefile. The only modification we made to it was to output the elapsed time after each iteration, which has negligible effect on performance.  

The data consists in this experiment is $N = 500,000$ data points in $d = 4$, drawn i.i.d from a Gaussian with identity covariance, and $K = 500$. With all optimisations (level 3) convergence is obtained within 20 seconds. We notice that each optimisation provides a significant boost to convergence speed. The faster initialisation at levels 2 and 3 is due to the fact that using inter-center distances allows nearests and second nearests to be determined with fewer distances and comparisons. 

Finally we note that the implementation of \cite{Kanungo_2002_local}, kmlocal, is about $100\times$ slower than our level 3 implememtation on this data. We have not run any other experiments comparing performance.

\begin{figure}
\begin{center}
\includegraphics[width=1.0\columnwidth]{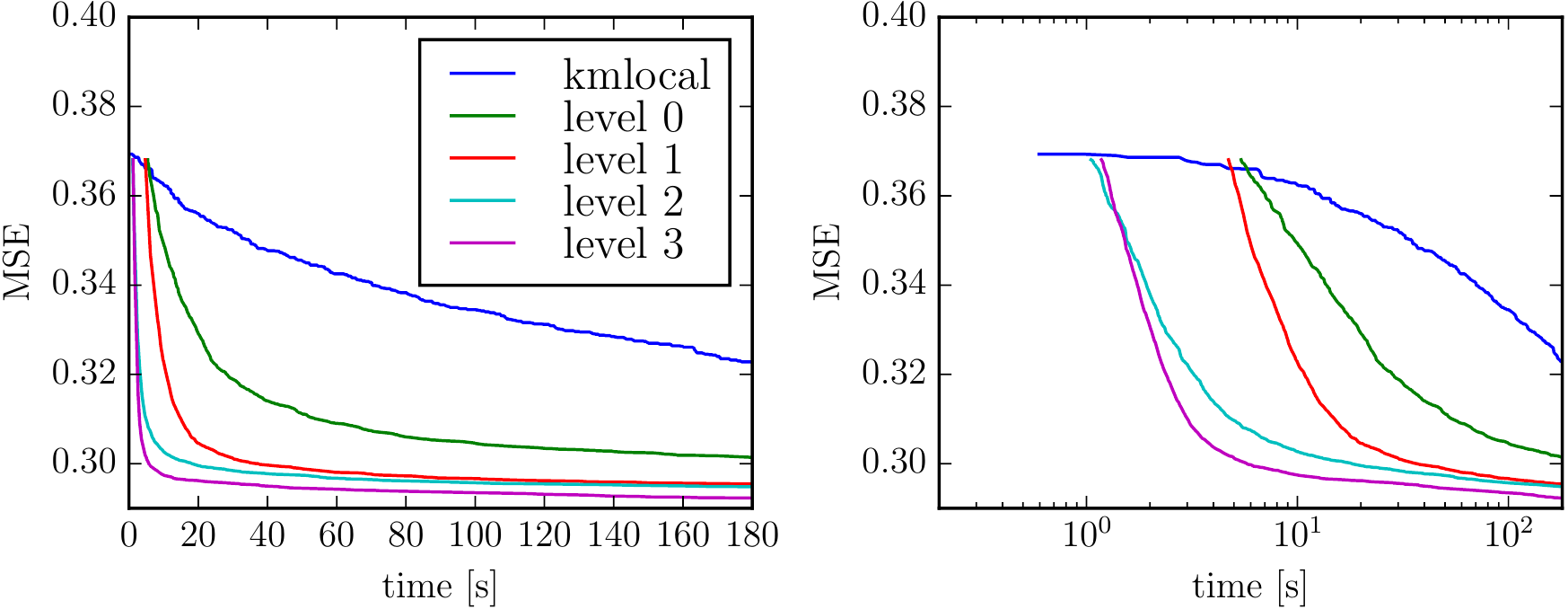}
\caption{Comparing the different optimisation levels and the implementation of \clarans{} of \cite{Kanungo_2002_local}, kmlocal at \protect\url{https://www.cs.umd.edu/~mount/Projects/KMeans/}. Left and right are the same but for a logarithmic scale for the time-axis on the right. The data being clustered here is $N=500,000$ elements in $d=4$, drawn from a Gaussian distribution with identity covariance, and $K = 500$. We see that the various levels of optimisation provide significant accelerations, and that the implementation in kmlocal is 2 orders of magnitude slower than our level 3 optimised implementation. }
\label{comparisona}
\end{center}
\end{figure}

\begin{figure}
\begin{center}
\includegraphics[width=1.0\columnwidth]{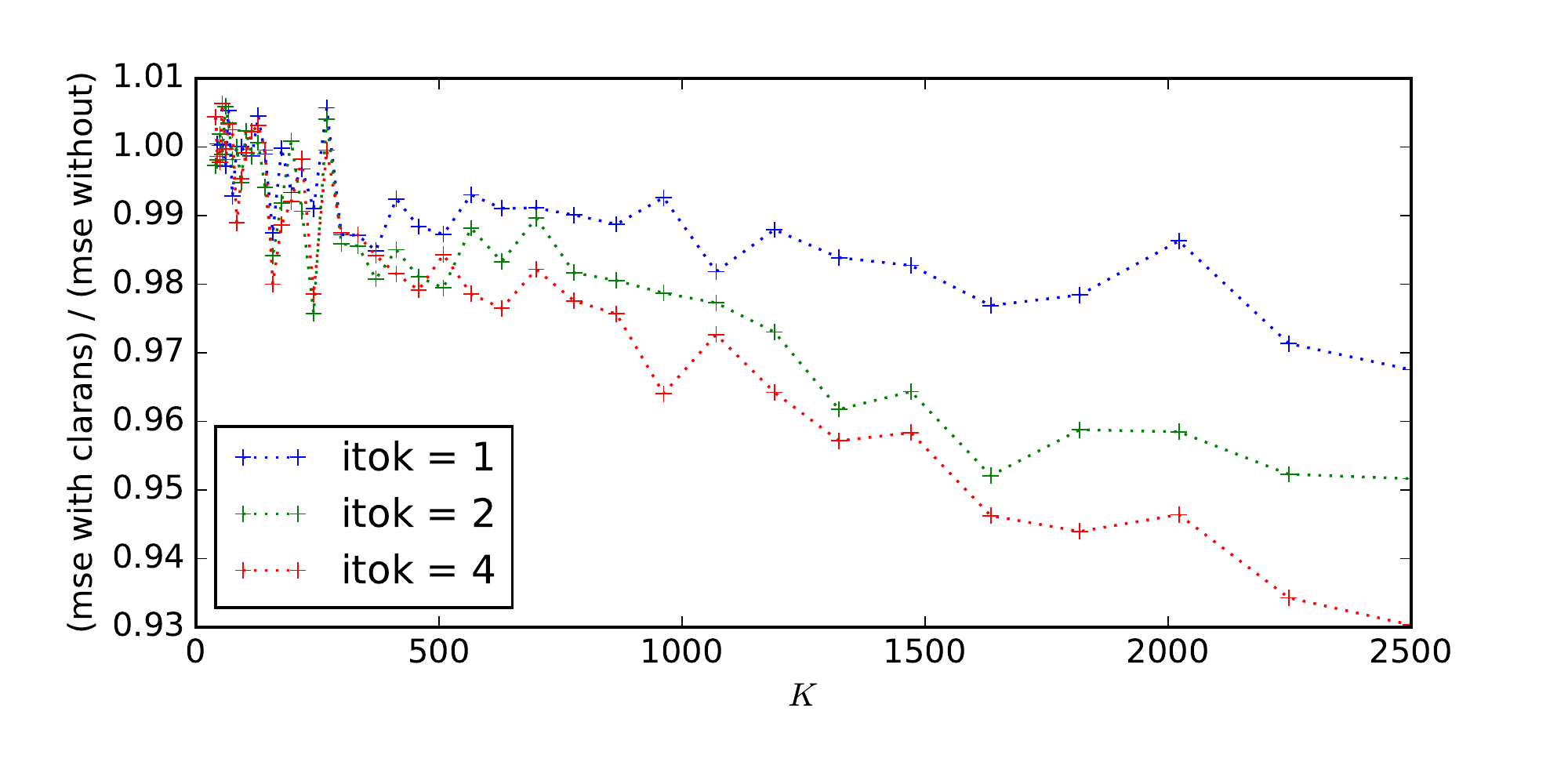}
\caption{Improvement obtained using \clarans{} for different values of $K$ (horizontal axis). The experimental setup is as follows. $N = 20,000$ points are drawn from a 3-D Gaussian with identity covariance. Then for each of $40$ values of $K$ on the horizontal axis, (1) \kmeanspp{} is run for fixed seed, and the time it takes to run is recorded (call it $T_{++}$). \clarans{} is then run for a mulitple `itok' of $T_{++}$, where `itok' is one of $\{0,1,2,4\}$. `itok' of $0$ corresponds to no \clarans{}. After \clarans{} has completed, \lloyd{} is run. For `itok' in ${1,2,4}$ the ratio of the final MSE with `itok' 0 (no \clarans{}) is plotted. This value is the fraction of the MSE without running \clarans{}. We see that the dependence of the improvement on $K$ is significant, with larger $K$ values benefitting more from \clarans{}. Also, as expected, larger `itok' results in lower MSE. }
\label{comparisonb}
\end{center}
\end{figure}

\end{document}